\documentclass[a4paper,12pt]{amsart}

\usepackage[latin1]{inputenc}
\usepackage{textcomp}
\usepackage{amsmath}
\usepackage{amssymb}
\usepackage{amsthm}
\usepackage{booktabs}
\usepackage{textcomp}
\usepackage{enumerate}
\usepackage[sort,numbers]{natbib}
\usepackage{array}
\usepackage{psfrag}
\usepackage{graphicx}
\usepackage{bbm}
\usepackage{ifthen}
\usepackage{verbatim} 
\usepackage{multirow}

\newcommand{\1}{\ensuremath{\mathbbm{1}}}

\newcommand{\as}{a.s.}

\newcommand{\eqsp}{\;}
\newcommand{\eqdef}{\ensuremath{\stackrel{\mathrm{def}}{=}}}

\usepackage{algorithmic}
\usepackage{algorithm}
\usepackage[novbox]{pdfsync}           
\usepackage[textwidth=4cm, textsize=footnotesize,disable]{todonotes}    
\usepackage{xargs}

\newcommand{\defeq}{\mathrel{\mathop:}=}

\newcommand{\F}{\mathcal{F}}

\def\Xset{\ensuremath{\mathsf{X}}}
\def\Xsigma{\ensuremath{\mathcal{X}}}
\def\eqsp{\,}
\newcommand{\uarg}{\,\cdot\,}
\newcommand{\ud}{\mathrm{d}}
\newcommand{\R}{\mathbb{R}}

\newcommand{\E}{\mathbb{E}}
\newcommand{\charfun}[1]{\mathbbm{I}\left\{#1\right\}}
\newcommand{\seq}[1]{\left\{#1_n\right\}}
\def\rmd{\mathrm{d}}

\newcommand{\chunk}[3]{\ensuremath{#1^{(#2:#3)}}}

\def\bbeta{\boldsymbol{\beta}}

\usepackage{aliascnt}
\usepackage{hyperref}
\usepackage{xargs}

\newtheorem{theorem}{Theorem}
\newaliascnt{proposition}{theorem}
\newtheorem{proposition}[proposition]{Proposition}
\aliascntresetthe{proposition}
\newaliascnt{lemma}{theorem}
\newtheorem{lemma}[lemma]{Lemma}
\aliascntresetthe{lemma}
\newaliascnt{corollary}{theorem}
\newtheorem{corollary}[corollary]{Corollary}
\aliascntresetthe{corollary}
\newaliascnt{definition}{theorem}

\aliascntresetthe{definition}
\newaliascnt{example}{theorem}

\aliascntresetthe{example}
\theoremstyle{remark}
\newaliascnt{remark}{theorem}
\newtheorem{remark}[remark]{Remark}
\aliascntresetthe{remark}

\newcommandx\prob[2][1=,2=]{{\mathbb P}_{#1}^{#2} }
\newcommandx\esp[2][1=,2=]{{\mathbb E}_{#1}^{#2} }
\newcommandx\cesp[4][1=,2=]{\ensuremath{{\mathbb E}_{#1}^{#2}\left[ \left. #3 \right| #4 \right]}}
\newcommandx\cprob[4][1=,2=]{\ensuremath{{\mathbb P}_{#1}^{#2}\left[ \left. #3 \right| #4 \right]}}
\newcommandx\sequence[2][2=n]{\ensuremath{\{ #1_{#2} \}_{#2 \geq 0}}}
\newcommandx\ball[3][1=]{\mathrm{B}^{#1} (#2,#3)}

\newcommand{\rset}{\ensuremath{\mathbb{R}}}
\newcommand{\nset}{\ensuremath{\mathbb{N}}}

\def\lleb{\ensuremath{\lambda^{\mathrm{Leb}}}}

\def\covmat{\Sigma}
\def\bcovmat{\boldsymbol{\Sigma}}
\def\brho{\boldsymbol{\rho}}
\def\bM{\mathbf{M}}
\def\bP{\mathbf{P}}
\def\bS{\mathbf{S}}
\def\bV{\mathbf{V}}
\def\bpi{\boldsymbol{\pi}}
\def\brho{\boldsymbol{\rho}}

\newcommand{\blazej}[1]{\todo[color=blue!20]{{\bf BM:} #1}}

\newcommand{\coint}[1]{\left[#1\right)}
\newcommand{\ocint}[1]{\left(#1\right]}
\newcommand{\ooint}[1]{\left(#1\right)}
\newcommand{\ccint}[1]{\left[#1\right]}
\newcommandx\positivematrixset[3][1=,3=]{%
\ifthenelse{\equal{#1}{}}{\mathcal{M}^{#3}_+(#2)}{\mathcal{M}^{#3}_+(#2,#1)}
}
\newcommand{\set}[1]{\mathsf{#1}}
\newcommandx\supnorm[2][1=]{\left\| #2 \right\|^{#1}_\infty}

\newcounter{hypA}
\newenvironment{hypA}{\begin{sf}\refstepcounter{hypA}\begin{itemize}
  \item[({\bf A\arabic{hypA}})]}{\end{itemize}\end{sf}}

\newcounter{hypB}

\newcommandx\A[2][1=]{
\ifthenelse{\equal{#1}{}}
{\hspace{-1mm}(\textbf{A\ref{#2}})\hspace{-1mm}}
{\hspace{-1mm}(\textbf{A\ref{#1}--\ref{#2}})\hspace{-1mm}}
}

\newcommandx\B[2][1=]{
\ifthenelse{\equal{#1}{}}
{\hspace{-1mm}(\textbf{B\ref{#2}})\hspace{-1mm}}
{\hspace{-1mm}(\textbf{B\ref{#1}--\ref{#2}})\hspace{-1mm}}
}

\def\bX{\mathbf{X}}
\def\bx{\mathbf{x}}
\def\bY{\mathbf{Y}}

\def\bh{\mathbf{h}}

\newcommand{\tvnorm}[1]{\ensuremath{\left\|#1\right\|_{\mathrm{TV}}}}
\newcommand{\Vnorm}[2]{\ensuremath{\left\|#1\right\|_{#2}}}
\newcommandx\functionset[2][1=]{\mathbb{F}_{#1}(#2)}

\begin{document}

\title{Adaptive parallel tempering algorithm}
\author{B\l{}a\.{z}ej Miasojedow, Eric
Moulines and Matti Vihola}

\address{B\l{}a\.{z}ej Miasojedow, LTCI, CNRS UMR 8151, Institut Mines-Télécom/Télécom ParisTech,
46, rue Barrault, 75634 Paris Cedex 13, France}
\address{Eric Moulines, LTCI, CNRS UMR 8151, Institut Mines-Télécom/Télécom ParisTech,
46, rue Barrault, 75634 Paris Cedex 13, France}
\address{Matti Vihola,
Department of Mathematics and Statistics,
P.O.Box 35,
FI-40014 University of Jyväskylä,
Finland}

\maketitle

\begin{abstract} 
Parallel tempering is a generic Markov chain Monte Carlo sampling
method which allows good mixing with multimodal target distributions,
where conventional Metropolis-Hastings algorithms often fail. The mixing
properties of the sampler depend strongly on the choice of tuning
parameters, such as the temperature schedule and the proposal
distribution used for local exploration.  We propose an adaptive
algorithm which tunes both the temperature schedule and the parameters
of the random-walk Metropolis kernel automatically. We prove the
convergence of the adaptation and a strong law of large numbers for
the algorithm. We illustrate the performance of our method with
examples. Our empirical findings indicate that the algorithm can cope
well with different kind of scenarios without prior tuning.
\end{abstract} 

\section{Introduction}
Markov chain Monte Carlo (MCMC) is a generic method to
approximate an integral of the form
\[
    I \eqdef \int_{\rset^d} f(x) \pi(x) \ud x \eqsp,
\]
where $\pi$ is a probability density function, which can be evaluated
point-wise up to a normalising constant. Such an integral occurs
frequently when computing Bayesian posterior expectations
\cite[e.g.,][]{robert-casella,gilks-mcmc}.

The random-walk Metropolis algorithm \cite{metropolis} works often
well, provided the target density $\pi$ is, roughly speaking,
sufficiently close to unimodal.
The efficiency of the Metropolis algorithm can be optimised
by a suitable choice of the proposal distribution. The proposal
distribution can be chosen automatically by several adaptive MCMC
algorithms; see
\cite{haario-saksman-tamminen-am,atchade-rosenthal,roberts-rosenthal-examples,andrieu-thoms}
and references
therein.

When $\pi$ has multiple well-separated modes, the random-walk based
methods tend to get stuck to a single mode for long periods of time,
which can lead to false convergence and severely erroneous
results. Using a tailored Metropolis-Hastings algorithm
can help, but in many cases finding a good
proposal distribution is difficult. Tempering of $\pi$, that is,
considering auxiliary distributions with density proportional to
$\pi^\beta$ with $\beta\in(0,1)$ often provides better mixing within
modes \cite{swendsen-wang,marinari-parisi,hansmann,woodard-schmidler-huber-rapid}.
We focus here particularly on the parallel tempering algorithm, which
is also known
as the replica exchange Monte Carlo and the Metropolis
coupled Markov chain Monte Carlo.

The tempering approach is particularly tempting in such settings where
$\pi$ admits a physical interpretation, and there is good intuition
how to choose the temperature schedule for the algorithm. In general,
choosing the temperature schedule is a non-trivial task, but there
are generic guidelines for temperature selection, based on both
empirical findings and theoretical analysis. First rule of thumb
suggests that the temperature progression should be (approximately)
geometric; see, e.g.~\cite{kofke}. Kone and Kofke linked also the mean
acceptance rate of the swaps \cite{kone-kofke}; this has been further
analysed by Atchadé, Roberts and Rosenthal
\cite{atchade-roberts-rosenthal}; see also
\cite{roberts-rosenthal-simulated-tempering}.

Our temperature adaptation is based on the latter; we try to optimise
the mean acceptance rate of the swaps between the chains in adjacent
temperatures. Our scheme has similarities with that proposed by
Atchadé, Roberts and Rosenthal~\cite{atchade-roberts-rosenthal}. The
key difference in our method is that we propose to adapt continuously
during the simulation. We show that the temperature adaptation
converges, and that the point of convergence is unique with mild and
natural conditions on the target distribution.

The local exploration in our approach relies on the random walk
Metropolis algorithm. The proposal distribution, or more precisely,
the scale/shape parameter of the proposal distribution, can be adapted
using several existing techniques like the covariance estimator
\cite{haario-saksman-tamminen-am} augmented with an adaptive scaling
pursuing a given mean acceptance rate
\cite{andrieu-thoms,atchade-fort,roberts-rosenthal-examples,atchade-rosenthal}
which is motivated by certain asymptotic results
\cite{roberts-rosenthal-scaling,roberts-gelman-gilks-scaling}.
It is also possible to use a robust shape estimate which 
enforces a given mean acceptance rate \cite{vihola-ram}.

We start by describing the proposed algorithm in Section
\ref{sec:alg}. Theoretical results on the convergence of the
adaptation and the ergodic averages are given next in Section
\ref{sec:theory}. In Section \ref{sec:experiments}, we illustrate the
performance of the algorithm with examples. The proofs of the
theoretical results are postponed to Section \ref{sec:proofs}.

\section{Algorithm}
\label{sec:alg}
\subsection{Parallel tempering algorithm} 
The parallel tempering algorithm defines a
Markov chain over the product
space $\Xset^L$, where $\Xset\subset\R^d$
\begin{equation}
\label{eq:definition-joint-process}
\bX_k = (X_k^{(1)},\ldots,X_k^{(L)}) =
(\chunk{X_k}{1}{L}) \eqsp.
\end{equation}
Each of the chains $X_k^{(\ell)}$ targets a `tempered' version of the
target distribution $\pi$.
Denote by $\bbeta= (\chunk{\beta}{1}{L})$ the  inverse temperature, which are such that  $1=\beta^{(1)}>\beta^{(2)}>\cdots>\beta^{(L)}>0$.
and by  $Z(\beta)$ the normalising constant
\begin{equation}
\label{eq:definition-normalizing-constant}
Z(\beta) \eqdef \int \pi^{\beta}(x) \ud x\eqsp,
\end{equation}
which is assumed to be finite.
The parallel tempering algorithm is constructed so that the Markov chain $\{ \bX_k \}_{k \geq 0}$
is reversible with respect to the product density
\begin{equation}
\label{eq:definition-bpi}
\bpi_{\bbeta}(x^{(1)},\ldots,x^{(L)})
\eqdef \frac{\pi^{\beta^{(1)}}(x^{(1)})}{Z(\beta^{(1)})}
\times \cdots \times \frac{\pi^{\beta^{(L)}}(x^{(L)})}{{Z(\beta^{(L)})}} \eqsp,
\end{equation}
over the product space $(\Xset^L, \Xsigma^{\otimes L})$.

Each time-step may be decomposed into two successive moves: the swap
move and the propagation (or update) move; for the latter, we consider
only random-walk Metropolis moves.

We  use the following notation to distinguish the state of the
algorithm after the swap step (denoted $\bar{\bX}_n$) and after the random walk step, or equivalently after
a complete step (denoted $\bX_n$). The state is then updated according to
\begin{equation}
\label{eq:definition-state-pt}
\bX_{n-1}\stackrel{\bS_{\bbeta}}{\longrightarrow}\bar{\bX}_{n-1}
    \stackrel{\bM_{(\bcovmat,\bbeta)}}{\longrightarrow} \bX_n,
\end{equation}
where the two kernels $\bM_{(\bcovmat,\bbeta)}$ and $\bS_{\bbeta}$ are respectively defined as
\begin{itemize}
\item $\bM_{(\bcovmat,\bbeta)}$ denotes the tensor product kernel  on the product space $\Xset^L$
\begin{equation}
\label{eq:product-kernel}
\bM_{(\bcovmat,\bbeta)}(\chunk{x}{1}{L};A_1\times \dots A_L)= \prod_{\ell=1}^L M_{(\covmat^{(\ell)},\beta^{(\ell)})}(x^{(\ell)},A_\ell)
\end{equation}
where each $M_{(\covmat^{(\ell)},\beta^{(\ell)})}$ is a random-walk
Metropolis kernel targeting $\pi^{\beta^{(\ell)}}$ with increment
distribution $q_{\covmat^{(\ell)}}$, where $q_{\covmat}$ is the
density of a multivariate Gaussian with zero mean and covariance
$\covmat$,
\begin{multline}
\label{eq:definitionMn}
M_{(\covmat,\beta)}(x,A) \eqdef \int_A \alpha_{\beta}(x,y) q_{\covmat}(y-x) \rmd y \\
+ \delta_x(A) \int (1 -  \alpha_{\beta}(x,y)) q_{\covmat} (y-x) \rmd y \eqsp,
\end{multline}
where
\begin{equation}
\label{eq:definition-alphan}
\alpha_{\beta}(x,y) \eqdef 1 \wedge \frac{\pi^\beta(y)}{\pi^\beta(x)}
\eqsp, \quad \text{for all $(x,y) \in \Xset \times \Xset$}\eqsp .
\end{equation}
In practical terms, $\bM_{(\bcovmat,\bbeta)}$ means that one applies a
random-walk Metropolis step separately for each of the chains
$\bar{X}_{n-1}^{(\ell)}$.
\item $\bS_{\bbeta}$ denotes the Markov kernel of the swap steps, targeting the product distribution
$\bpi_{\bbeta} \propto \pi^{\beta^{(1)}}\otimes\cdots\otimes
\pi^{\beta^{(L)}}$,
\begin{multline}
\label{eq:definitionSn}
\bS_{\bbeta}(\chunk{x}{1}{L};A)= \frac{1}{L-1} \sum_{j=1}^{L-1} \varpi^{(j)}_{\bbeta}(x^{(j)},x^{(j+1)}) J^{(j)}(x^{(j)},x^{(j+1)};A)  \\
+ \frac{1}{L-1} \sum_{j=1}^{L-1}( 1-\varpi^{(j)}_{\bbeta}(x^{(j)},x^{(j+1)})  ) \; \delta_{\chunk{x}{1}{L}}(A) \eqsp,
\end{multline}
where $\varpi^{(j)}_{\bbeta}$ is the probability of accepting a swap between levels $j$ and $j+1$, which is given by
\begin{equation}
\label{eq:probabilite-swap}
\varpi^{(j)}_{\bbeta} (x^{(j)},x^{(j+1)}) \eqdef
1 \wedge   \left( \frac{\pi (x^{(j+1)})}{\pi(x^{(j)})} \right)^{\beta^{(j)}- \beta^{(j+1)}} \eqsp,
\end{equation}
and
\begin{multline}
\label{eq:definitionJ}
J^{(j)}(x^{(j)},x^{(j+1)};A) \\
\eqdef \idotsint_A \delta_{x^{(j)}} (\rmd y^{(j+1)}) \delta_{x^{(j+1)}}(\rmd y^{(j)}) \prod_{i \in \{1,\dots,L\} \setminus \{j,j+1\}} \delta_{x^{(i)}}(\rmd y^{(i)}) \eqsp.
\end{multline}
The above defined swap step means choosing a random index
$\ell\in\{1,\ldots,L-1\}$ uniformly, and then proposing to swap the
adjacent states, $\bar{X}_{n-1}^{(\ell+1)} = X_{n-1}^{(\ell)}$
and $\bar{X}_{n-1}^{(\ell)} = X_{n-1}^{(\ell+1)}$.
and accepting this swap with probability given in
\eqref{eq:probabilite-swap}.
\end{itemize}

\subsection{Adaptive parallel tempering algorithm} 
In the adaptive version of the parallel tempering algorithm, the temperature parameters are continuously updated along the run of the algorithm. We denote the sequence of inverse temperatures
\begin{equation}
\label{eq:inverse-temperatures}
\left\{ \bbeta_n \right\}_{n \geq 0} = \left\{  \chunk{\beta_n}{1}{L} \right\}_{n\ge 0}
\end{equation}
which are parameterised by the vector-valued process
\begin{equation}
\label{eq:definition-rho}
\left\{ \brho_n \right\}_{n \ge 0} \eqdef \left\{ \chunk{\rho_n}{1}{L-1} \right\}_{n\ge 0} \eqsp,
\end{equation}
through $\beta_n^{(1)} \eqdef 1$
and $\beta_n^{(\ell)}= \beta^{(\ell)}(\chunk{\rho}{1}{\ell-1}_n)$ for $\ell \in \{2, \dots, L\}$ with
\begin{equation}
\beta^{(\ell+1)}(\chunk{\rho}{1}{\ell}) \eqdef
\beta^{(\ell)}(\chunk{\rho}{1}{\ell-1}) \exp(-\exp(\rho^{(\ell)}))  \eqsp .
\label{eq:beta-rho-def}
\end{equation}

Because the inverse temperatures are adapted at each iteration,
the target distribution of the chain changes from step to step as well.
Our adaptation of the temperatures is performed using the following stochastic
approximation procedure
\begin{equation}
\label{def:temp_adapt}
\rho^{(\ell)}_{n} =
\Pi_{\rho} \left( \rho^{(\ell)}_{n-1} + \gamma_{n,1} H^{(\ell)}( \chunk{\rho_{n-1}}{1}{\ell},\bX_n) \right) \eqsp,
\qquad 1\le \ell \le L-1 \eqsp,
\end{equation}
where $(\chunk{\rho_n}{1}{L-1})$ is defined in \eqref{eq:beta-rho-def},
$\Pi_{\rho}$ is the projection onto the constraint set
$\ccint{\underline{\rho},\overline{\rho}}$, which will be discussed
further in Section \ref{sec:implementation}.
Moreover,
\begin{align}
\label{eq:definition-H-ell}
&H^{(\ell)}(\chunk{\rho}{1}{\ell},\bx)=
1\wedge \left(\frac{\pi(x^{(\ell+1)})}{\pi(x^{(\ell)})} \right)^{\Delta \beta^{(\ell)}(\chunk{\rho}{1}{\ell})}
-\alpha^* \eqsp, \\
\label{eq:definition-Delta-beta}
&\Delta \beta^{(\ell)} (\chunk{\rho}{1}{\ell})= \beta^{(\ell)}(\chunk{\rho}{1}{\ell-1}) - \beta^{(\ell+1)}(\chunk{\rho}{1}{\ell}) \eqsp.
\end{align}

We will show in Section \ref{sec:theory} that the algorithm is
designed in such a way that the inverse temperatures converges to a
value for which the mean probability of accepting a swap move between
any adjacent-temperature chains is constant and is equal to $\alpha^*$.

We will also adapt the random-walk proposal distribution at each level.
We describe below another possible algorithm for performing such a task.
In the theoretical part, for simplicity,
we will consider only on with the seminal adaptive
Metropolis algorithm \cite{haario-saksman-tamminen-am} augmented with
scaling adaptation
\cite[e.g.][]{roberts-rosenthal-examples,atchade-fort,andrieu-thoms}.
In this algorithm, we estimate  the covariance matrix of the target distribution at each temperature and rescale it
to control the acceptance ratio at each level in stationarity.

Define by $\positivematrixset{d}$ the set of $d\times d$ positive
definite matrices. For $A\in\positivematrixset{d}$, we denote by
$\varrho_{\min}(A)$ and $\varrho_{\max}(A)$ the smallest and the
largest eigenvalues of $A$, respectively. For $\varepsilon \in(0,1]$,
define by $\positivematrixset[\varepsilon]{d} \subset
\positivematrixset{d}$ the convex subset
\begin{equation}
\label{eq:definition-positive-matrix-set}
\positivematrixset[\varepsilon]{d} \eqdef \left\{ \covmat \in \positivematrixset{d}: \varepsilon \leq \varrho_{\min}(\covmat) \leq \varrho_{\max}(\covmat) \leq \varepsilon^{-1} \right\} \eqsp.
\end{equation}
The set $\positivematrixset[\varepsilon]{d}$ is a compact subset of the open cone of positive definite matrices.

We denote by $\Gamma_n^{(\ell)}$ the current estimate of the covariance at level $\ell$, which is updated as follows
\begin{align}
\Gamma_n^{(\ell)} &=\Pi_{\Gamma}\left[ (1-\gamma_{n,2}) \Gamma_{n-1}^{(\ell)}  + \gamma_{n,2}
      (X_n^{(\ell)}-\mu_{n-1}^{(\ell)}){t(X_n^{(\ell)}-\mu_{n-1}^{(\ell)})} \right]\label{eq:am-cov}, \\
   \mu_n^{(\ell)} &= (1-\gamma_{n,2})\mu_{n-1}^{(\ell)} + \gamma_{n,2}
   X_n^{(\ell)} \eqsp,\label{eq:am-mean}
\end{align}
where $t(\cdot)$ is the matrix transpose and 
$\Pi_{\Gamma}$ is the projection on to the set
$\positivematrixset[\varepsilon]{d}$;
see Section \ref{sec:implementation}.
The scaling parameters is updated so that the acceptance rate in
stationarity converges to the target $\alpha^\star$,
\begin{equation}
\label{eq:scaling-adaptation}
    T_n^{(\ell)} =\Pi_T\left( T_{n-1}^{(\ell)} +\gamma_{n,3} \left[ \left(1 \wedge \frac{\pi^{\beta^{(\ell)}_{n-1}}(Y^{(\ell)}_n)}{\pi^{\beta^{(\ell)}_{n-1}}(\bar{X}_{n-1}^{(\ell)})} \right) -\alpha^* \right]\right)\eqsp,
\end{equation}
where $\Pi_T$ is is the projection onto
$\ccint{\underline{T},\overline{T}}$; see Section
\ref{sec:implementation} and $Y_n^{(\ell)}$ is the proposal at level $\ell$, assumed to be conditionally independent from the past draws
and distributed according to a Gaussian with mean $\bar{X}_{n-1}^{(\ell)}$ and covariance  matrix $\covmat^{(\ell)}_{n-1}$ which is given by
\begin{equation}
\label{eq:definition-covmat_n}
\covmat^{(\ell)}_n = \exp(T_n^{(\ell)}) \, \Gamma_n^{(\ell)} \eqsp.
\end{equation}
In the sequel we denote by $\bY_n$ the vector of proposed moves at time step $n$,
\begin{equation}
\label{eq:definition-Y}
\{ \bY_n \}_{n \ge 0} = \{ \chunk{Y_n}{1}{L}  \}_{n \geq 0} \eqsp.
\end{equation}


\subsection{Alternate random-walk adaptation} 

In order to reduce the number of parameters in the adaptation
especially in higher dimensions, we propose to use a common covariance
for all the temperatures, but still employ separate scaling.
More specifically,
\begin{align}
    \Gamma_n
  &=
  (1-\gamma_{n,2})\Gamma_{n-1}+
      \frac{\gamma_{n,2}}{L}
      \sum_{\ell=1}^L
      (X_n^{(\ell)}-\mu_{n-1})t(X_n^{(\ell)}-\mu_{n-1})
      \eqsp ,
      \label{eq:global-mean-def}
      \\
   \mu_n &= (1-\gamma_{n,2})\mu_{n-1} + \frac{\gamma_{n,2}}{L}
      \sum_{\ell=1}^L X_n^{(\ell)} \eqsp ,
      \label{eq:global-cov-def}
\end{align}
and set $\Sigma_n^{(\ell)} = \exp(T_n^{(\ell)}) \Gamma_n$.

Another possible implementation of the random-walk adaption, robust adaptive
Metropolis (RAM) \cite{vihola-ram},
is defined by a single dynamic adjusting the covariance parameter and
attaining a given acceptance rate.
Specifically, one recursively finds a lower-diagonal matrix
$\Gamma_n^{(\ell)}\in\rset^{d\times d}$ with positive diagonal
satisfying
\begin{equation}
    \Gamma_{n}^{(\ell)}t({\Gamma}_{n}^{(\ell)}) = \Gamma_{n-1}^{(\ell)}
    \left[I +
        \gamma_{n,2}(\alpha_n -
        \alpha^*)
        u(Z_n^{(\ell)})t\big(u(Z_n^{(\ell)})\big)
\right]
        t({\Gamma}_{n-1}^{(\ell)})\eqsp,
   \label{eq:ram-def}
\end{equation}
where $Z_n^{(\ell)} \eqdef Y_n^{(\ell)}-\bar{X}_{n-1}^{(\ell)}$ and
$u(x) \defeq \charfun{x\neq 0} x/|x|$, 
and let  $\Sigma_n^{(\ell)} = \Gamma_n^{(\ell)} t({\Gamma}_n^{(\ell)})$.

The potential benefit of using this estimate
instead of \eqref{eq:am-cov}--\eqref{eq:scaling-adaptation} is that RAM
finds, loosely speaking, a `local' shape of the target distribution,
which is often in practice close to a convex combination of the shapes of
individual modes. In some situations, this proposal shape might allow better
local exploration than the global covariance shape.


\subsection{Implementation details}
\label{sec:implementation} 

In the experiments, we use the desired acceptance rate
$\alpha^*=0.234$ suggested by theoretical results for the swap kernel
\cite{kone-kofke,atchade-roberts-rosenthal}
and for the random-walk Metropolis kernel
\cite{roberts-rosenthal-scaling,roberts-gelman-gilks-scaling}.
We employ the step size sequences
$\gamma_{n,i} = c_i(n+1)^{-\xi_i}$ with constants
$c_1,c_3\in(0,\infty)$ and $c_2\in(0,1]$ and
$\xi_1,\xi_2,\xi_3\in(1/2,1)$. This is a common choice in the
stochastic approximation literature.

The projections $\Pi_\rho$, $\Pi_\Gamma$ and $\Pi_T$ in
\eqref{def:temp_adapt}, \eqref{eq:am-cov} and
\eqref{eq:scaling-adaptation}, respectively, are used to enforce the
stability of the adaptation process in order to simplify theoretical
analysis of the algorithm.  We have not observed instability
empirically, and believe that the algorithm would be stable without
projections; in fact, for the random-walk adaptation, there exist some
stability results \cite{saksman-vihola,vihola-asm,vihola-collapse}.
Therefore, we recommend setting the limits in the constraint sets as
large as possible, within the limits of numerical accuracy.

It is possible to employ other strategies for proposing swaps of the
tempered states. Specifically, it is possible to try more than one
swap at each iteration, even go through all the temperatures, without
changing the invariant distribution of the chain. We made some
preliminary tests with other strategies, but the results were not
promising, so we decided to keep the common approach of a
single randomly chosen swap.

In the temperature adaptation, it is also possible to enforce the
geometric progression, and only adapt one parameter. More
specifically, one can use $\rho^{(\ell)}_n \defeq \rho_n$ for all
$\ell\in\{1,\ldots,L-1\}$ and perform the adaptation
\eqref{def:temp_adapt} to update $\rho_{n-1}\to \rho_n$. This strategy
might induce more stable behaviour of the temperature parameter
especially when the number of levels is high. On the other hand, it 
can be dangerous because the asymptotic acceptance probability across
certain temperature levels can get low, inducing poor mixing. 

We consider only Gaussian proposal distributions in the random-walk
Metropolis kernels. It is possible to employ also other proposals; in
fact our theoretical results extend directly for example to
the multivariate Student proposal distributions.

We note that the adaptive parallel tempering algorithm can be used
also in a block-wise manner, or in Metropolis-within-Gibbs framework.
More precisely, the adaptive random-walk chains can be run as
Metropolis-within-Gibbs, and the state swapping can be done in the
global level. This approach scales better with respect to the
dimension in many situations. Particularly, when the model is
hierarchical, the structure of the model can allow significant
computational savings. Finally, it is straightforward to extend the
adaptive parallel tempering algorithm described above to general
measure spaces. For the sake of exposition, we present the algorithm 
only in $\R^d$.


\section{Theoretical results}
\label{sec:theory} 

\subsection{Formal definitions and assumptions} 

Denote by $\seq{Y}$ the proposals of the random-walk Metropolis step. 
We define the following filtration
\begin{equation}
\mathcal{F}_n=\sigma\left\{ \bX_0, (\bX_k, \bar{\bX}_{k-1}, \bY_{k-1}),  k=1,\dots,n,\; \right\} \eqsp.
\end{equation}%
By construction, the covariance matrix $\bcovmat_n \eqdef (\chunk{\covmat_n}{1}{L})$ and the parameters $\bbeta_n \eqdef \chunk{\beta_n}{1}{L}$ are adapted to the filtration $\mathcal{F}_n$.
With these notations and assumptions, for any time step $n \in \nset$,
\begin{equation*}
\cprob{\bX_{n+1}\in\uarg}{\mathcal{F}_n} =\int \bS_{\bbeta_n}(\bX_n,\ud z )\bM_{(\bcovmat_n,\bbeta_n)}(z, \uarg)  = \bS_{\bbeta_n} \bM_{(\bcovmat_n,\bbeta_n)}(\bX_n,\uarg)
\end{equation*}
Therefore, denoting $\bP_{(\bcovmat_n,\bbeta_{n})}\eqdef \bS_{\bbeta_n} \bM_{(\bcovmat_n,\bbeta_n)}$, we get
\begin{equation}
\label{def:expn}
\cesp{f(\bX_{n+1})}{\mathcal{F}_n}=\bP_{(\bcovmat_n,\bbeta_{n})}f(\bX_n) \eqsp,
\end{equation}
for all $n \in \nset$ and all bounded measurable functions $f: \Xset^L \to \rset$.

We will consider the following assumption on the target distribution,
which ensures a geometric ergodicity of a random walk Metropolis chain
\cite{andrieu-moulines,jarner-hansen}. Below, $|\cdot|$ applied to a
vector (or a matrix) stands for the Euclidean norm.
\begin{hypA}
    \label{a:super-exp} 
    The density $\pi$ is bounded, bounded away from zero
    on compact sets, differentiable, such that
\begin{align}
    \label{eq:super-exp}
    \lim_{r\to\infty} \sup_{|x|\ge r}
    \frac{x}{|x|} \cdot \nabla \log \pi(x) &= -\infty \\
    \lim_{r\to\infty} \sup_{|x|\ge r}
    \frac{x}{|x|} \cdot \frac{\nabla\pi(x)}{|\nabla\pi(x)|} &< 0.
    \label{eq:reg-contour}
\end{align}
\end{hypA} 

In words, \A{a:super-exp} only requires that the target
distribution is sufficiently regular, and the tails decay at a rate
faster than exponential. We remark that the tempering approach is
only well-defined when $\pi^\beta$ are integrable with exponents of
interest $\beta>0$---this is the case always with \A{a:super-exp}.


\subsection{Geometric ergodicity and continuity of parallel tempering
  kernels} 

We first state and prove that the parallel tempering algorithm is
geometrically ergodic under \A{a:super-exp}. This result might be of
independent interest, because geometric ergodicity is well known to
imply central limit theorems.

We show that, under mild conditions, this
kernel is phi-irreducible, strongly aperiodic, and $V$-uniformly
ergodic, where the function $V$ is the sum of an appropriately chosen
negative power of the target density.
Specifically, for $\beta \in \rset_+$, consider the following drift function
\begin{equation}
\label{eq:lyapunov-function-junct}
\bV_\beta (\chunk{x}{1}{L}) \eqdef \sum_{\ell=1}^L V_{\beta}
(x^{(\ell)}) \eqsp,
\end{equation}
where for $x \in \Xset$,
\begin{equation}
\label{eq:lyapunov-function}
V_\beta (x)= (\pi(x)/\supnorm{\pi})^{-\beta/2} \eqsp.
\end{equation}
For $\beta_0 > 0$, define the set
\begin{equation}
\label{eq:definition-K-beta}
\mathcal{K}_{\beta_0} \eqdef
\left\{ \chunk{\beta}{1}{L} \in \ocint{0,1}^L,
\beta_0 \leq \beta^{(L)} \leq \dots \leq \beta^{(1)} \right\} \eqsp.
\end{equation}
We denote the $V$-variation of a signed measure $\mu$ as $\|\mu\|_V
\defeq \sup_{f:|f|\le V} \mu(f)$, where the supremum is taken over all
measurable functions $f$. The $V$-norm of a function is defined as
$\|f\|_V \eqdef \sup_x |f(x)|/V(x)$. 

\begin{theorem}
\label{theo:geometric-ergodicity-parallel-tempering}
Assume \A{a:super-exp}. Let $\epsilon > 0$ and $\beta_0 > 0$.
 Then there exists $C_{\epsilon,\beta_0}<\infty$ and $\varrho_{\epsilon,\beta_0}<1$ such that, for all $\bx \in \Xset^L$,  $\bcovmat \in \positivematrixset[\epsilon]{d}$ and $\bbeta \in \mathcal{K}_{\beta_0}$,
\begin{equation}
    \Vnorm{\bP^n_{(\bcovmat,\bbeta)}(\bx,\cdot)-\bpi_{\bbeta}}{\bV_{\beta_0}}\leq
    C_{\epsilon,\beta_0}\varrho^n_{\epsilon,\beta_0}\bV_{\beta_0}(x).
    \label{eq:geom-erg}
\end{equation}
\end{theorem}

Geometric ergodicity in turn implies the existence of a solution of the Poisson equation, and also provides bounds on the growth of this solution \cite[Chapter~17]{Meyn2009}
\begin{corollary}\label{coro:Poisson}Assume \A{a:super-exp}.
Let $\epsilon > 0$ and $\beta_0 > 0$.
For any measurable function $f$ with $\Vnorm{f}{\bV_{\beta_0}^\alpha}<\infty$ for some $\alpha \in \ocint{0,1}$ there exists a unique (up to an additive constant) solution of the Poisson equation
\begin{equation}\label{eq:Poisson_equation}
g -\bP_{(\bcovmat,\bbeta)} g =f-\bpi_{\bbeta}(f)\eqsp.
\end{equation}
This solution is denoted $\hat{f}_{(\bcovmat,\bbeta)}$. In addition, there exists a constant $D_{\epsilon,\beta_0} < \infty$ such that
\begin{equation}
\label{eq:Vnorm_Poisson}
\Vert \hat{f}_{(\bcovmat,\bbeta)}\Vert_{\bV_{\beta_0}^\alpha}\leq
D_{\epsilon,\beta_0}\Vert f \Vert_{\bV_{\beta_0}^\alpha}\eqsp.
\end{equation}
\end{corollary}

We will next establish that the parallel tempering kernel is locally
Lipshitz continuous.
For any $\beta > 0$, denote by
$D_{\bV_{\beta}}\left[(\bcovmat,\bbeta),(\bcovmat',\bbeta')\right]$
the $\bV_{\beta}$-variation of the kernels $\bP_{(\bcovmat,\bbeta)}$
and $\bP_{(\bcovmat',\bbeta')}$,
\begin{equation}
\label{supp:eq:defi:DVnorm}
D_{\bV_{\beta}}\left[(\bcovmat,\bbeta),(\bcovmat',\bbeta')\right]
\eqdef \sup_{\bx \in \Xset^L} \frac{\Vnorm{\bP_{(\bcovmat,\bbeta)}(\bx,\cdot) - \bP_{(\bcovmat',\bbeta')}(\bx,\cdot)}{\bV_{\beta}} }{\bV_{\beta}(\bx)}\eqsp.
\end{equation}
For $\beta_0 \in \ooint{0,1} $ and $\eta > 0$, define the set
\begin{equation}
\label{eq:definition-K-beta-eta}
\mathcal{K}_{\beta_0,\eta} \eqdef
\left\{ \beta_0 \leq \beta^{(L)} \leq \dots \leq \beta^{(1)} \leq 1 \eqsp, \beta^{(k)} - \beta^{(k+1)} \geq \eta \right\} \eqsp.
\end{equation}

\begin{theorem}\label{theo:continuity-P}
Assume \A{a:super-exp}. Let $\epsilon > 0$, $\beta_0 > 0$ and $\eta > 0$.
For any $\alpha \in \ocint{0,1}$, there exists
$K_{\epsilon,\alpha,\beta_0,\eta}<\infty$ such that, for any
$\bcovmat,\bcovmat' \in \positivematrixset[\epsilon]{d}[L]$ and any
$\bbeta,\bbeta' \in \mathcal{K}_{\beta_0,\eta}$, it holds that
\[
D_{\bV^\alpha_{\beta_0}}\left[(\bcovmat,\bbeta),(\bcovmat',\bbeta')\right]\leq K_{\epsilon,\alpha,\beta_0,\eta}\left\{|\bbeta-\bbeta'|+|\bcovmat-\bcovmat'|\right\}\eqsp.
\]
\end{theorem}


\subsection{Strong law of large numbers} 

We can state an ergodicity result for 
the adaptive parallel tempering algorithm, given the step size
sequences satisfy the following natural condition.

\begin{hypA}\label{a:step-size}
 Assume that the step sizes  $\{\gamma_{n,i}, n \in \nset \}$, $i=1,2,3$ defined in \eqref{def:temp_adapt},\eqref{eq:am-cov}, and \eqref{eq:scaling-adaptation} are non-negative and satisfy following conditions
\begin{enumerate}[(i)]
\item  For $i=1,2,3$,  $\sum_{n\geq1}\gamma_{n,i} =\infty$ and $\sum_{n\geq1}n^{-1}\gamma_{n,i} < \infty$
\label{a:step-size-item1}
\item \label{a:step-size-item2}
$\sup_{n\in\nset}\gamma_{n,2}\leq 1$
\end{enumerate}
\end{hypA}
\begin{remark}
    \label{rem:step-size-example}
It is easy to see that $\gamma_{n,i} = c_i(n+1)^{-\xi_i}$ with some
$c_1,c_3\in(0,\infty)$ and $c_2\in(0,1]$ and
$\xi_1,\xi_2,\xi_3\in(0,1]$ satisfy \A{a:step-size}.
\end{remark}

\begin{theorem}
\label{theo:law-of-large-numbers}
Assume \A{a:super-exp}-\A{a:step-size} and $\E
\left[ \bV_{\beta_0}(\bX_0) \right] <\infty$. Then, for any function
$f: \Xset^L \to \rset$ such that $\Vnorm{f}{\bV_{\beta_0}^\alpha}<\infty$ for some
$\alpha \in \ooint{0,1}$ and given $\lim_{n \to \infty} \bpi_{\bbeta_n}(f)$ exists,
we have
\[\frac{1}{n}\sum_{i=1}^n f(\bX_i) \longrightarrow \lim_{n \to \infty} \bpi_{\bbeta_n} (f)\quad \text{a.s.}\]
\end{theorem}
\begin{remark}
    \label{rem:cool-limit}
In practice, one is usually only interested in integrating with
respect to $\pi$, which means functions $f$ depending only on the
first coordinate, that is, $f(\chunk{x}{1}{L})= f(x^{(1)})$. In this
case, the limit condition is trivial, because 
$\bpi_{\bbeta_n}(f)= \pi(f)$ for all $n\in\nset$.
\end{remark}


\subsection{Convergence of temperature adaptation}
\label{sec:adapt-conv} 

The strong law of large numbers
(\autoref{theo:law-of-large-numbers}) does not require the convergence
of the inverse temperatures, if only the coolest chain $x^{(1)}$ is
involved (\autoref{rem:cool-limit}). It is, however, important to work out the convergence of the
adaptation, because then we know what to expect on the asymptotic
behaviour of the algorithm. Having the convergence, it is also
possible to establish central limit theorems \cite{andrieu-moulines};
however, we do not pursue it here.

We denote the associated mean field of the stochastic approximation procedure
\eqref{def:temp_adapt} by
\[
\bh(\brho) \eqdef
\left[h^{(1)}(\rho^{(1)}),\ldots,h^{(L-1)}(\rho^{(1)},\dots,\rho^{(L-1)})\right] \eqsp,
\]
where
\[
h^{(\ell)}(\rho^{(1)},\dots,\rho^{(\ell)})\eqdef
\int H^{(\ell)}(\rho^{(1)}, \dots, \rho^{(\ell)},\bx) \bpi_{\bbeta}(\ud \bx) \eqsp.
\]
We may write
\[
h^{(\ell)}(\brho)=\iint
  \varpi_{\bbeta}^{(\ell)}(x^{(\ell)},x^{(\ell+1)})
  \frac{\pi^{\beta^{(\ell)}}(\ud x^{(\ell)})}{Z(\beta^{(\ell)})}
  \frac{\pi^{\beta^{(\ell+1)}}(\ud x^{(\ell+1)})}{Z(\beta^{(\ell+1)})}-\alpha^*,
\]
where $Z(\beta)$ is the normalising constant defined in
\eqref{eq:definition-normalizing-constant}.

The following result establishes the existence and uniqueness of the
stable point of the adaptation. In words, the following result implies
that there exist unique temperatures so that the mean rate of
accepting proposed swaps is $\alpha^*$.

\begin{proposition}
    \label{prop:uniqueness-temperature} 
    Assume \A{a:super-exp}. Then, there exists a
    unique $\hat{\brho} \eqdef [ \hat{\rho}^{(1)},\dots,\hat{\rho}^{(L-1)}]$
    solution of the system of equations $h^{(\ell)}(\rho^{(1)},\dots, \rho^{(\ell)})=0$, $\ell \in \{1,\dots, L-1\}$.
\end{proposition} 

\begin{remark} 
In Proposition \autoref{prop:uniqueness-temperature}, it is sufficient
to assume that the support of $\pi$ has infinite Lebesgue measure and
that $\int \pi^\kappa(x) \ud x < \infty$ for all $0<\kappa\le 1$; see
\autoref{lem:monotonic:h}.
\end{remark} 

\begin{remark} 
In case the support of $\pi$ has a finite Lebesgue measure, it is not
difficult to show that for a sufficiently large number of levels $L\ge
L_0$ there is no solution $\hat{\brho}$. On the contrary, in formal terms,
$\hat{\rho}^{(\ell)} = \infty$ for $\ell \ge L_0$, so that the
corresponding inverse temperatures $\hat{\beta}^{(\ell)} = 0$ for
$\ell \ge L_0+1$. For our algorithm, this would imply that
it simulates asymptotically $\pi^0/Z(0)$, the
uniform distribution on the support of $\pi$, with the levels $\ell\ge
L_0+1$.
\end{remark}

For the convergence of the temperature adaptation, we require more stringent 
conditions on the step size sequence.
\begin{hypA}\label{a:step-size-v2}
 Assume that step sizes $\{ \gamma_{n,i}, n \in \nset \}$ defined in \eqref{def:temp_adapt},\eqref{eq:am-cov},\eqref{eq:am-mean} and \eqref{eq:scaling-adaptation} are non-negative and satisfy following conditions
\begin{enumerate}[(i)]
\item $\sum_{n \geq 1} \gamma_{n,i} = \infty$, $\sum_{n\geq 1} \gamma_{n,1}^2 < \infty$,
and $\sum_{n \geq 1} \gamma_{n,1} \gamma_{n,j} < \infty$, $j=2,3$.
\label{a:step-size-v2-item1}
\item $\sup_{n\in\nset}\gamma_{n,2}\leq 1$
\label{a:step-size-v2-item2}
\item $\sum_{n\geq1}|\gamma_{n+1,1}-\gamma_{n,1}|<\infty$
\label{a:step-size-v2-item3}
\end{enumerate}
\end{hypA}
It is easy to check that the sequences introduced in 
\autoref{rem:step-size-example} satisfy \A{a:step-size-v2} if we
assume in addition that $\xi_1,\xi_2,\xi_3\in(1/2,1]$.

\begin{theorem}\label{theo:convergence-temp-adapt}
 Assume \A{a:super-exp}-\A{a:step-size-v2}, $\E\bV_{\beta_0}(\bX_0)<\infty$. In addition for all $\ell=1,\dots,L-1$ we assume that
 $\underline{\rho}<\hat{\rho}^{(\ell)}<\overline{\rho}$, where $\hat{\brho}$ is given by \autoref{prop:uniqueness-temperature}. Then
\[
\lim_{n\to\infty} \brho_n= \hat{\brho}\quad \as\eqsp .
\]
\end{theorem}


\section{Experiments}
\label{sec:experiments}

We consider two different type of examples: mixture of Gaussians in
Section \ref{sec:mixture} and a challenging spatial imaging example in
Section \ref{sec:imaging}. In all the experiments, we use the step
size sequences $\gamma_{n,\uarg} = (n+1)^{-0.6}$, except for RAM
adaptation, where $\gamma_{n,2} = \min\{0.9, d\cdot (n+1)^{-0.6}\}$
(see \cite{vihola-ram} for a discussion). We did not observe numerical
instability issues, so the adaptations were not enforced to be stable
by projections. We used the following initial values for the adapted
parameters: temperature difference $\rho_0^{(\ell)} = 1$, covariances
$\Sigma_0^{(\ell)} = I$ and scalings $\theta_0^{(\ell)}=1$.

\subsection{Mixture of Gaussians}
\label{sec:mixture}

We consider first a well-known two-dimensional mixture of Gaussians example
\cite[e.g.][]{liang-wong,baragatti-grimaud-pommeret}.  The example
consists of 20 mixture components with means in $[0,10]^2$ and each
component has a diagonal covariance $\sigma^2 I$, with
$\sigma^2=0.01$.  \autoref{fig:points-baragatti} shows an example of
the points simulated by our parallel tempering algorithm in this
example, when we use $L=5$ energy levels and the default (covariance)
adaptation to adjust the random walk proposals.
\autoref{fig:baragatti-temperatures} shows the convergence of the
temperature parameters in the same example.
\begin{figure}
\includegraphics[width=\textwidth]{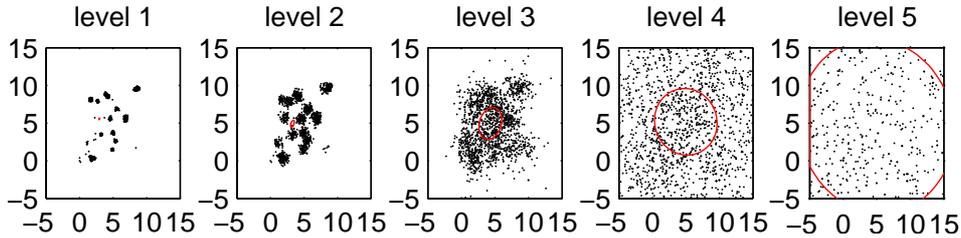}
\caption{Simulated points of the tempered distributions over 5000
iterations. The random-walk proposal is illustrated as an ellipsoid.}
\label{fig:points-baragatti}
\end{figure}
\begin{figure}
\includegraphics[width=\textwidth]{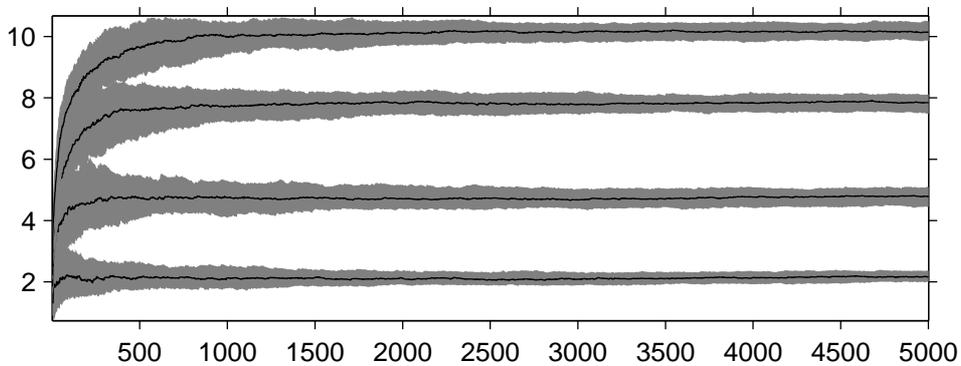}
\caption{Convergence of the log-temperatures in the mixture example.
The 10\%--90\% quantiles over 100 runs of the algorithm are drawn in grey
and the black line shows the median.}
\label{fig:baragatti-temperatures}
\end{figure}

We computed estimates of the means and the squares
of the coordinates with $N=5000$ iterations of which $2500$ burn-in,
and show the mean and standard deviation (in parenthesis) over
100 runs of our parallel tempering algorithm in \autoref{tab:baragatti}.
We considered three different random-walk adaptations:
the default adaptation in
\eqref{eq:am-cov}--\eqref{eq:scaling-adaptation} (Cov),
with common mean and covariance estimators as defined in
\eqref{eq:global-mean-def}--\eqref{eq:global-cov-def}
(Cov(g)) and the RAM update defined in \eqref{eq:ram-def}.
\autoref{tab:baragatti} shows the results in the same form as
\cite[Table 3]{baragatti-grimaud-pommeret} to allow easy comparison.
When comparing with \cite{baragatti-grimaud-pommeret},
our results show smaller deviation than the unadapted parallel
tempering, but bigger deviation than their samplers including also
equi-energy moves. We remark that we did not adjust our algorithm at
all for the example, but let the adaptation take care of that.
There are no significant differences between the
random-walk adaptation algorithms.
\begin{table}
\caption{The test of \cite{baragatti-grimaud-pommeret}, Table 3 case
(A) with $L=5$ temperature levels, 5000 iterations and 2500 burn-in.}
\label{tab:baragatti}
\begin{tabular}{lcccc}
\toprule
           &  $\E[X_1]$ & $\E[X_2]$ & $\E[X_1^2]$ & $\E[X_2^2]$ \\
True value & 4.478 & 4.905 & 25.605 & 33.920 \\
\midrule
Cov & 4.469 (0.588) & 4.950 (0.813) & 25.329 (5.639) & 34.209
(8.106) \\
Cov(g) & 4.389 (0.537) & 4.803 (0.692) & 24.677 (5.411) & 32.865
(6.660) \\
RAM & 4.530 (0.524) & 4.946 (0.811) & 26.111 (5.308) & 34.331 (8.292) \\
\bottomrule
\end{tabular}
\end{table}

When looking the simulated points in \autoref{fig:points-baragatti},
it is clear that three temperature levels is enough to allow good
mixing in the above example. We repeated the example with $L=3$ energy
levels, and increased the number of iterations to $N=8333$ in order to
have a comparable computational cost.
The summary of the results in
\autoref{tab:baragatti2} indicates increased 
accuracy than with $L=5$ levels,
and the accuracy comes close to the results reported in
\cite{baragatti-grimaud-pommeret}
for samplers with equi-energy moves.
\begin{table}
\caption{The test of \cite{baragatti-grimaud-pommeret}, Table 3 case
(A) with $L=3$ temperature levels, 8333 iterations and 4167 burn-in.}
\label{tab:baragatti2}
\begin{tabular}{lcccc}
\toprule
           &  $\E[X_1]$ & $\E[X_2]$ & $\E[X_1^2]$ & $\E[X_2^2]$ \\
True value & 4.478 & 4.905 & 25.605 & 33.920 \\
\midrule
Cov & 4.480 (0.416) & 4.957 (0.571) & 25.542 (4.164) & 34.420
(5.669) \\
Cov(g) & 4.488 (0.422) & 4.884 (0.551) & 25.719 (4.190) & 33.520
(5.476) \\
RAM & 4.490 (0.407) & 4.881 (0.541) & 25.667 (4.281) & 33.622 (5.631) \\
\bottomrule
\end{tabular}
\end{table}

We considered also a more difficult modification of the mixture example
above. We decreased the variance of the mixture components to
$\sigma^2=0.001$ and increased the dimension to $d=8$. The mixture
means of the added coordinates were all zero. We ran our adaptive
parallel tempering algorithm in this case with $L=8$ temperature
levels;
\autoref{tab:baragatti-harder} summarises the results
with different number of iterations. In all the cases, the first half
of the iterations were burn-in. In this scenario, the different
random-walk adaptation algorithms have slightly different behaviour.
Particularly, the common mean and covariance estimates (Cov(g))
seem to improve over the separate covariances (Cov). The RAM approach
seems to provide the most accurate results. However, we believe that
this is probably due to the special properties of the example, namely
the fact that all the mixture components have
a common covariance, and the RAM converges close to this in the first
level; see the discussion in \cite{vihola-ram}.

\begin{table}
\caption{The root mean square errors in the
modified mixture example with $\sigma^2=0.001$,
$d=8$ and $L=8$.}
\label{tab:baragatti-harder}
\begin{tabular}{llccc}
\toprule
$N$ & Est. &  Cov & Cov(g) & RAM \\
\midrule
\multirow{2}{*}{$10k$} & $\E[X]$ &
3.080 & 2.245 & 1.660 \\
& $\E[|X|^2]$ &
27.577 & 20.426 & 16.428 \\
\multirow{2}{*}{$20k$} & $\E[X]$ &
1.788 & 1.580 & 1.429 \\
& $\E[|X|^2]$ &
18.577 & 15.712 & 14.475 \\
\multirow{2}{*}{$40k$} & $\E[X]$ &
1.439 & 1.267 & 0.952 \\
& $\E[|X|^2]$ &
15.471 & 12.769 & 9.364 \\
\multirow{2}{*}{$80k$} & $\E[X]$ &
1.257 & 0.975 & 0.698 \\
& $\E[|X|^2]$ &
13.017 & 9.414 & 6.981 \\
\multirow{2}{*}{$160k$} & $\E[X]$ &
1.096 & 0.680 & 0.508  \\    
& $\E[|X|^2]$ & 
11.093 & 7.038 & 5.122  \\
\bottomrule
\end{tabular}
\end{table}

\subsection{Spatial imaging}
\label{sec:imaging}
As another example,
we consider identifying ice floes from polar satellite images as
described by Banfield and Raftery~\cite{Banfield1992}. The image under
consideration is a 200 by 200 gray-scale satellite image, and
we focus on the same 40 by 40 subregion as in \cite{LukeBornn2011}.
The goal is to
identify the presence and position of polar ice floes. Towards this
goal, Higdon \cite{higdon} employs a Bayesian model with an Ising
model prior and following posterior distribution on
$\{0,1\}^{40\times40}$,
\[
\log(\pi(x|y))\propto\sum_{1\leq i,j,\leq 40}\alpha\1\{x_{i,j}=y_{i,j}\} +\sum_{(i,j)\sim(i',j')}\beta\1\{x_{i,j}=x_{i',j'}\}\eqsp,
\]
where neighbourhood relation ($\sim$) is defined as vertical, horizontal and
diagonal adjacencies of each pixel. Posterior distribution favours $x$
which are similar to original image $y$ (first term) and for which
the neighbouring pixels are equal (second term).

In \cite{higdon} and
\cite{LukeBornn2011}, the authors observed that standard MCMC algorithms
which propose to flip one pixel at a time fail to explore the modes of
the posterior. There are, however, some advantages of using such an
algorithm, given we can overcome the difficulty in mixing
between the modes. Specifically, in order to compute (the log-difference of)
the unnormalised density values, we need only to explore
the neighbourhoods of the pixels that have changed.
Therefore, the proposal with one pixel flip at a time has a low
computational cost. Moreover, such an algorithm is easy to implement.

We used our parallel tempering algorithm with the above mentioned proposal
with $L=10$ temperature levels to simulate the posterior of this model
with parameters $\alpha=1$ and $\beta=0.7$.
We ran 100 replications of $N=100000$ iterations of the algorithm.
Obtained result are shown in \autoref{fig:ising}
is similar to \cite{higdon} and
\cite{LukeBornn2011}. We emphasize again that our algorithm provided
good results without any prior tuning.
\begin{figure}
\begin{center}
\includegraphics[width=\textwidth]{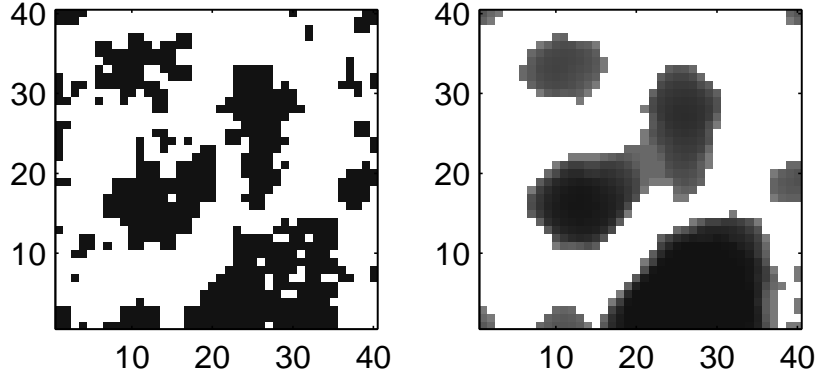}
\caption{Spatial model example: original image (left),
posterior estimate based on 100 replications of adaptive parallel
tempering (right)}
\label{fig:ising}
\end{center}
\end{figure}

\section{Proofs}
\label{sec:proofs}
\subsection{Proof of
\autoref{theo:geometric-ergodicity-parallel-tempering}} 
The proof follows by arguments in the literature that guarantee a 
geometric ergodicity for the individual random-walk Metropolis
kernels, and by observing that the swap kernel is invariant under
permutation-invariant functions.

We start with an easy lemma showing that a drift in cooler chain
implies a drift in the higher-temperature chain.

\begin{lemma}\label{lem:PWinequality} 
Consider the drift function $W \eqdef c\pi^{-\kappa}$ for some positive constants $\kappa$ and $c$. Then, for any $\covmat \in \positivematrixset{d}$,
\[
\beta \leq \beta' \Longrightarrow M_{(\covmat,\beta^\prime)}W(x)\leq
M_{(\covmat,\beta)} W(x) \eqsp, \quad \text{for all $x \in \Xset$
\eqsp.}
\]
\end{lemma} 
\begin{proof} 
We write
\begin{multline*}\frac{M_{(\covmat,\beta)} W(x)}{W(x)}=\int_{\{y\;:\; \pi(y)\geq\pi(x)\}} \left( \frac{\pi(x)}{\pi(y)} \right)^\kappa \; q_\covmat(x-y)\ud y \\+\int_{\{y\;:\; \pi(y)<\pi(x)\}}\left[1-\left(\frac{\pi(y)}{\pi(x)}\right)^\beta+\left(\frac{\pi(y)}{\pi(x)}\right)^{\beta-\kappa}\right]q_\covmat(x-y)\ud y
\end{multline*}
First term is independent on $\beta$, since $\beta\mapsto 1- a^\beta+a^{\beta-\kappa}$ for $a\in[0,1]$ is non-increasing the second term is also non-increasing with respect to $\beta$.
\end{proof} 

To control the ergodicity of each individual random-walk Metropolis
sampler, it is required to have a control on the minorisation and
drift constants for the kernels $M_{(\covmat,\beta)}$. The following
proposition provides such explicit control.
\begin{lemma}\label{lem:small-set-RWM} 
Assume \A{a:super-exp}. Let $\varepsilon > 0$ and $\beta > 0$.
There exist $\lambda_{\varepsilon,\beta} \in \coint{0,1}$ and $b_{\varepsilon,\beta} < \infty$ such that for any $\covmat \in \positivematrixset[\varepsilon]{d}$, we get
\begin{equation}
\label{eq:drift-condition}
M_{(\covmat,\beta)} V_\beta \leq \lambda_{\varepsilon,\beta} V_\beta + b_{\varepsilon,\beta}  \eqsp,
\end{equation}
where $V_\beta \eqdef (\pi/\supnorm{\pi})^{-\beta/2}$, where $\supnorm{\pi}= \sup_{x \in \Xset} \pi(x)$.
\end{lemma} 
\begin{proof} 
It is easily seen that if the target distribution is super-exponential
in the tails \A{a:super-exp}, then all the tempered versions
$\pi^\beta/ Z(\beta)$, where the normalising constant
$Z(\beta)$ is defined in \eqref{eq:definition-normalizing-constant},
satisfy \A{a:super-exp} as well.

The result then follows from \citet[Proposition~12]{andrieu-moulines}.
\end{proof} 

\begin{proposition} \label{prop:drift} 
Assume \A{a:super-exp} and let $\epsilon > 0$ and $\beta_0 > 0$. 
Then, there exists $\lambda_{\epsilon,\beta_0}<1$, and $b_{\epsilon,\beta_0}<\infty$
such that, for all $\bcovmat \in \positivematrixset[\epsilon]{d}[L]$ and
$\bbeta \in \mathcal{K}_{\beta_0}$, 
\begin{align}
&\label{eq:la-preuve-pour-bM}
\bM_{(\bcovmat,\bbeta)} \bV_{\beta_0} \leq \lambda_{\epsilon,\beta_0} \bV_{\beta_0}+ b_{\epsilon,\beta_0} \eqsp, \\
&\label{eq:invariance-de-bS}
\bS_{\bbeta} \bV_{\beta_0}= \bV_{\beta_0} \eqsp, \\
&\label{eq:la-preuve-pour-bP}
\bP_{(\bcovmat,\bbeta)} \bV_{\beta_0} \leq \lambda_{\epsilon,\beta_0} \bV_{\beta_0}+ b_{\epsilon,\beta_0} \eqsp,
\end{align}
\end{proposition} 
\begin{proof} 
By \autoref{lem:PWinequality}, since $\bbeta \in \mathcal{K}_{\beta_0}$, we get
\begin{align}
\label{eq:MV_inequality}
\bM_{(\bcovmat,\bbeta)}\bV_{\beta_0}(\chunk{x}{1}{L})
&=\sum_{\ell=1}^L
M_{(\covmat^{(\ell)},\beta^{(\ell)})}V_{\beta_0}(x^{(\ell)}) \\
&\leq \sum_{\ell=1}^L
M_{(\covmat^{(\ell)},\beta_0)}V_{\beta_0}(x^{(\ell)}) \eqsp.\nonumber
\end{align}
Then, by \autoref{lem:small-set-RWM}, since $\bcovmat \in \positivematrixset[\epsilon]{d}$, it holds
\[
\sum_{\ell=1}^L M_{(\covmat^{(\ell)},\beta_0)}V_{\beta_0}(x^{(\ell)}) \leq \lambda_{\epsilon,\beta_0}
\sum_{\ell=1}^L V_{\beta_0}(x^{(\ell)}) + Lb_{\epsilon,\beta_0}  \eqsp.
\]
Thanks to the definition of the swapping kernel
\eqref{eq:definitionSn}-
-\eqref{eq:definitionJ},
for any positive measurable function $F: \Xset^L \to \rset_+$ which is
invariant by permutation\footnote{$F(\chunk{x}{1}{L})=
F(\chunk{x}{\sigma(1)}{\sigma(L)})$ for any $(\chunk{x}{1}{L}) \in
\Xset^L$ and any permutation $\sigma$ over the set $\{1,\dots,L\}$},
we get $\bS_{\bbeta} F(\chunk{x}{1}{L})= F(\chunk{x}{1}{L})$. Since
the drift function $\bV_{\beta_0}$ defined in
\eqref{eq:lyapunov-function-junct} is invariant by permutation we
obtain
\begin{align*}
\bS_{\bbeta} \left[\lambda_{\epsilon,\beta_0}\bV_{\beta_0}(\chunk{x}{1}{L})+Lb_{\epsilon,\beta_0}\right] &=\lambda_{\epsilon,\beta_0}\bS_{\bbeta}\bV_{\beta_0}(\chunk{x}{1}{L})+Lb_{\epsilon,\beta_0}
\\
&=\lambda_{\epsilon,\beta_0}\bV_{\beta_0}(\chunk{x}{1}{L})+Lb_{\epsilon,\beta_0} \eqsp.
\qedhere
\end{align*}
\end{proof} 

\begin{proposition} \label{prop:small_set} 
Assume \A{a:super-exp}. Let $\epsilon > 0$,
$\beta_0 > 0$ and $r>1$, and consider the level set
$\set{C}_r\eqdef\{\bx\in\Xset^L\eqsp:\eqsp \bV_{\beta_0}(\bx)\leq
r\}$. There exists a constant $\delta_{r,\epsilon,\beta_0} > 0$ such
that for all $\bcovmat \in
\positivematrixset[\epsilon]{d}$ and $\bbeta \in
\mathcal{K}_{\beta_0}$, the set $\set{C}_r$ is a
$(1,\delta_{r,\epsilon,\beta_0},\nu_r)$-small set for
$\bP_{(\bcovmat,\bbeta)}$, that is, 
\begin{equation}
\label{eq:smallset}
\bP_{(\bcovmat,\bbeta)}(\bx,\cdot)\geq \delta_{r,\epsilon,\beta_0}
\nu_{r,\epsilon,\beta_0}(\cdot)\quad \bx\in \set{C}_r,
\end{equation}
where $\nu_{r}(\cdot)= \lleb(\cdot \cap \set{C}_r)/\lleb(\set{C}_r)$
is a probability measure on $\set{C}_r$ and $\lleb$ stands for the
Lebesgue measure.
\end{proposition} 
\begin{proof} 
It is easy to see that the set $\set{C}_r$ is absorbing for $\bS_{\bbeta}$
because $\bS_{\bbeta} V_{\bbeta}(x)=
V_{\bbeta}(x)$ as observed in the proof of
\autoref{prop:small_set}, implying $\bS_{\bbeta}(x,\set{C}_r)=1$ for all
$x\in \set{C}_r$. Hence for $x\in \set{C}_r$
\begin{equation}
\label{eq:small_set_eq_1}
\bP_{(\bcovmat,\bbeta)}(\bx,\set{A})
\geq\int_{\set{C}_r\cap \set{A}}\prod_{\ell=1}^L \left(1\wedge
\frac{\pi(y^{(\ell)})}{\pi(x^{(\ell)})}\right)^{\beta^{(\ell)}}q_{\Sigma^{(\ell)}}(y^{(\ell)}-x^{(\ell)})
\rmd y^{(1:L)}\eqsp,
\end{equation}
where $q_{\covmat}$ is the multivariate Gaussian density with zero
mean and covariance $\covmat$.
Since the set $\set{C}_r$ is compact and $\bcovmat \in
\positivematrixset[\epsilon]{\beta_0}[L]$, there exists a
constant $C_{r,\epsilon,\beta_0}>0$ such that for any
$\ell=1,\dots,L$
\[
\inf_{(x,y) \in \set{C}_r \times \set{C}_r} q^{(\ell)}(y^{(\ell)}-x^{(\ell)}) \geq C_{r,\epsilon,\beta_0} \eqsp.
\]
Therefore, by \eqref{eq:small_set_eq_1} 
and since $\beta^{(\ell)}\in(0,1]$, 
we obtain for $x\in \set{C}_r$
\begin{align*}
\bP_{(\bcovmat,\bbeta)}(\bx,\set{A})
&\geq C_{r,\epsilon,\beta_0}^L \int_{\set{C}_r\cap \set{A}}\prod_{\ell=1}^L \left(1\wedge \frac{\pi(y^{(\ell)})}{\pi(x^{(\ell)})}\right)
\rmd y^{(1:L)}\eqsp.
\end{align*}
If $y= (\chunk{y}{1}{L})\in \set{C}_r$, we get $(\pi(y^{(\ell)})/\supnorm{\pi})^{-\beta_0} \leq r/L$ for all $\ell \in \{1,\dots,L\}$,
which implies that $(L/r)^{2/\beta_0} \supnorm{\pi} \leq \pi(y)$. Hence, for all $(x,y) \in C_r \times C_r$,
$\pi(y)/\pi(x) \leq (L/r)^{2/\beta_0}$, showing that
\begin{equation*}
\bP_{(\bcovmat,\bbeta)}(\bx,\set{A}) \geq C_{r,\epsilon,\beta_0}^L \left[(1 \wedge (L/r))^{2/\beta_0}\right]^L \lleb(C_r) \lleb_{C_r}(\set{C}_r\cap \set{A})
\eqsp,
\end{equation*}
where $\lleb_{\set{C}_r}(\cdot)= \lleb(\set{C}_r \cap \cdot)/ \lleb(\set{C}_r)$.
\end{proof} 

\begin{proof}[Proof of
\autoref{theo:geometric-ergodicity-parallel-tempering}] 
Choose a sufficiently large $r>1$ so that there exists a
$\tilde{\lambda}_{\epsilon,\beta_0}<1$ such that
\[
\bP_{(\bcovmat,\bbeta)} \bV_{\beta_0}(x) 
\leq \lambda_{\epsilon,\beta_0} \bV_{\beta_0}(x)
+ \charfun{x\notin C} b_{\epsilon,\beta_0}
\]
by \autoref{prop:drift}, 
where $\set{C}_r$ is defined in
\autoref{prop:small_set}.
This drift inequality, with the minorisation inequality in 
\autoref{prop:small_set} imply
$V$-uniform ergodicity 
\eqref{eq:geom-erg}
with constants depending only on
$\tilde{\lambda}_{\epsilon,\beta_0}<1$, 
$b_{\epsilon,\beta_0}$ and $ \delta_{r,\epsilon,\beta_0}$
\cite[e.g.][]{meyn-tweedie-computable}.
\end{proof}


\subsection{Proof of \autoref{theo:continuity-P}} 
We preface the proof of this Theorem by several technical lemmas.

\begin{lemma}\label{lem:sup-z-beta-beta} 
    For all $ (\beta,\beta') \in \ooint{0,1}^2 $ we have
\[ 
    \sup_{z \in
      \ccint{0,1}}|z^\beta-z^{\beta'}|\leq\frac{1}{\max\{\beta,
    \beta'\}}|\beta'-\beta|\eqsp. 
\]
\end{lemma} 
\begin{proof} 
 Without loss of generality, assume that $0<\beta<\beta'<1$.
 The function $w: \ccint{0,1}\to \rset$ defined as
 $w(z)= z^\beta-z^{\beta'}$ is continuous, non-negative and 
 $w(0)=w(1)=0$. Therefore, the maximum of this function is obtained 
 inside the interval $(0,1)$. By computing the derivative 
 $w'(z)=\beta z^{\beta-1}-\beta' z^{\beta'-1}$ and setting $w'(z)=0$,
 we find the maximum at 
 $z=\big(\frac{\beta}{\beta'}\big)^{1/(\beta'-\beta)}$, so
\begin{multline*}
\sup_{z \in \ccint{0,1}}
|z^\beta-z^{\beta'}|
=\left(\frac{\beta}{\beta'}\right)^{\frac{\beta}{\beta'-\beta}}\left(1-\frac{\beta}{\beta'}\right)
\leq\left(1-\frac{\beta}{\beta'}\right)=\frac{1}{\beta'}(\beta'-\beta)\eqsp.
\qedhere
\end{multline*}
\end{proof}
\begin{lemma}\label{lem:continuity-M-temp}
Set $\beta_0 \in \ocint{0,1}$.
There exists a constant $K_{\beta_0} < \infty$ such that for any $(\beta,\beta') \in \ccint{\beta_0,1}^2$, any covariance matrix $\covmat\in \positivematrixset{d}$,
\begin{multline}
\label{eq:bound-integral}
\int_\Xset \left(V_{\beta_0}(y)+V_{\beta_0}(x)\right)\left|\alpha_\beta(x,y)-\alpha_{\beta'}(x,y)\right|q_\covmat(y-x)
\lleb(\ud y) \\
\leq K_{\beta_0} |\beta - \beta'| V_{\beta_0}(x) \eqsp.
\end{multline}
In addition, 
for any measurable function $g$ with $\Vnorm{g}{V_{\beta_0}}\leq1$,
\begin{equation}
\label{eq:bound-difference-M}
\left|M_{(\covmat,\beta)}g(x)-M_{(\covmat,\beta')}g(x)\right|\leq K_{\beta_0}
|\beta-\beta'|V_{\beta_0}(x)\eqsp.
\end{equation}
\end{lemma}
\begin{proof}
Without loss of generality, we assume that $\beta<\beta'$.
Recall that $V_{\beta_0}(y) \propto \pi^{-\beta_0/2} (y)$.
Note that
\begin{multline*}
\int_\Xset V_{\beta_0}(y)\left|\alpha_\beta(x,y)-\alpha_{\beta'}(x,y)\right|q_\covmat(y-x)
\ud y \\
=
V_{\beta_0}(x)\int_{\set{R}_x} \left|\left(\frac{\pi(y)}{\pi(x)}\right)^{\beta-\beta_0/2}-\left(\frac{\pi(y)}{\pi(x)}\right)^{\beta'-\beta_0/2}\right|q_\covmat(y-x)
\ud y \eqsp,
\end{multline*}
where $\set{R}_x\defeq \{y\in \Xset: \pi(y)<\pi(x)\}$, and 
\begin{multline*}
V_{\beta_0}(x) \int_\Xset \left|\alpha_\beta(x,y)-\alpha_{\beta'}(x,y)\right|q_\covmat(y-x)
\ud y \\
=
V_{\beta_0}(x)\int_{\set{R}_x} \left|\left(\frac{\pi(y)}{\pi(x)}\right)^\beta-\left(\frac{\pi(y)}{\pi(x)}\right)^{\beta'}\right|q_\covmat(y-x)
\ud y \eqsp.
\end{multline*}
Using \autoref{lem:sup-z-beta-beta}, we get
\[
\int_{\set{R}_x} \left|\left(\frac{\pi(y)}{\pi(x)}\right)^{\beta-\beta_0/2}-\left(\frac{\pi(y)}{\pi(x)}\right)^{\beta'-\beta_0/2}\right|q_\covmat(y-x)
\ud y  
\leq \frac{1}{\beta-\beta_0/2} |\beta' - \beta| \eqsp,
\]
and similarly
\[
\int_{\set{R}_x} \left|\left(\frac{\pi(y)}{\pi(x)}\right)^\beta-\left(\frac{\pi(y)}{\pi(x)}\right)^{\beta'}\right|q_\covmat(y-x)
\ud y \leq \frac{1}{\beta} |\beta' - \beta| \eqsp,
\]
which concludes the proof of \eqref{eq:bound-integral}.

We consider now \eqref{eq:bound-difference-M}. Note that
\begin{align}\label{eq:M-M}
&\left|M_{(\covmat,\beta)}g(x)-M_{(\covmat,\beta')}g(x)\right|\\\nonumber
&=\left|\int_\Xset \left(g(y)-g(x)\right)\left(\alpha_\beta(x,y)-\alpha_{\beta'}(x,y)\right)q_\covmat(y-x)
\ud y \right|\\\nonumber
&\leq\int_\Xset \left|g(y)-g(x)\right|\left|\alpha_\beta(x,y)-\alpha_{\beta'}(x,y)\right|q_\covmat(y-x)
\ud y \\\nonumber
&\leq\int_\Xset \left(V_{\beta_0}(y)+V_{\beta_0}(x)\right)\left|\alpha_\beta(x,y)-\alpha_{\beta'}(x,y)\right|q_\covmat(y-x)
\ud y\nonumber
\end{align}
and we conclude using \eqref{eq:bound-integral}.
\end{proof} 

The following result is a restatement of 
\cite[Proposition 12 and (the proof of) Lemma 13]{andrieu-moulines}.
\begin{lemma}\label{lem-continuity-M-covmat} 
    For any $\epsilon>0$ 
there exists $K_\epsilon<\infty$ such that for any
$(\covmat,\covmat')\in\positivematrixset[\epsilon]{d}$, 
$\beta\in[0,1]$, and function $g$ with $\Vert
g\Vert_{V_{\beta_0}}\leq1$, we have
\[\left|M_{(\covmat,\beta)}g(x)-M_{(\covmat',\beta)}g(x)\right|\leq K_\epsilon|\covmat-\covmat'|V_{\beta_0}(x)\eqsp.\]
In addition there exists $K_q<\infty$ such that
\begin{equation}\label{eq:qcovmat}
\int_\Xset|q_\covmat(z)-q_{\covmat'}(z)| \ud z \leq K_q|\covmat-\covmat'|.
\end{equation}
\end{lemma} 
\autoref{lem-continuity-M-covmat} can be generalised also 
for non-Gaussian proposal distributions, including the multivariate
Student \cite[Appendix B]{vihola-asm}.

Now we show the local Lipschitz-continuity of the mapping
$(\bcovmat,\bbeta) \mapsto \bM_{(\bcovmat,\bbeta)}$.
\begin{proposition} 
\label{prop:continuity-bM}
Assume \A{a:super-exp}. Let $\epsilon > 0$ and $\beta_0 > 0$, and $\eta > 0$.
There exists a constant $K_{\epsilon,\beta_0,\eta} < \infty$ such that such that
for all $\bcovmat,\bcovmat' \in \positivematrixset[\epsilon]{d}[L]$,
$\bbeta,\bbeta' \in \mathcal{K}_{\beta_0,\eta}$, and
functions $g$ with $\Vnorm{g}{\bV_{\beta_0}}\leq1$,  we have
\[\left|\bM_{(\bcovmat,\bbeta)}g(\bx)-\bM_{(\bcovmat',\bbeta')}g(\bx)\right|\leq K_{\epsilon,\beta_0,\eta}\left\{|\bbeta-\bbeta'|+|\bcovmat-\bcovmat'|\right\}\bV_{\beta_0}(\bx)\eqsp.\]
\end{proposition} 
\begin{remark} 
In this proof, it is possible to set $\eta= 0$. The use of $\eta > 0$ is required in the proof of continuity.
\end{remark} 
\begin{proof} 
    We may write
\begin{multline}
\label{eq:definition-R}
\left|\bM_{(\bcovmat,\bbeta)}g(\bx)-\bM_{(\bcovmat',\bbeta')}g(\bx)\right|
\leq\left|\bM_{(\bcovmat,\bbeta)}g(\bx)-\bM_{(\bcovmat,\bbeta')}g(\bx)\right|\\
+\left|\bM_{(\bcovmat,\bbeta')}g(\bx)-\bM_{(\bcovmat',\bbeta')}g(\bx)\right|
\eqdef R_1(\bx)+R_2(\bx) \eqsp.
\end{multline}
First, we consider $R_1$. We prove by induction that  there exists a constant $K_{k,\epsilon,\beta_0,\eta}<\infty$ such that,
for all measurable $g$
such that $\Vnorm{g}{\bV_{\beta_0}} \leq 1$,
\begin{equation}\label{eq:prodM-M}
\left|\bM^{(k)}_{(\bcovmat,\bbeta)}g(\bx)-\bM^{(k)}_{(\bcovmat,\bbeta')}g(\bx)\right|
\leq K_{k,\epsilon,\beta_0,\eta}|\bbeta-\bbeta'|\bV_{\beta_0}(\bx) \eqsp,
\end{equation}
where 
$\bM^{(k)}_{(\bcovmat,\bbeta)}\eqdef\prod_{1\leq\ell\leq k}
\check{M}_{(\covmat^{(\ell)},\beta^{(\ell)})}$, and 
\[
    \check{M}_{(\covmat^{(\ell)},\beta^{(\ell)})}(x^{(1:L)},A_1\times\cdots\times
A_L)
    \eqdef
    M_{(\covmat^{(\ell)},\beta^{(\ell)})}(x^{(\ell)}, A_\ell)\prod_{i\neq\ell}
      \delta_{x^{(i)}}(A_i).
\]
We first establish the result for $k=1$. For any $\ell \in \{1,\dots,L\}$, we get
\begin{align*}
&\left|\check{M}_{(\covmat^{(\ell)},\beta^{(\ell)})}g(\bx)-
\check{M}_{(\covmat^{(\ell)},\beta'^{(\ell)})}g(\bx)\right|\\
\nonumber
&\leq\int_\Xset \left|g(\chunk{x}{1}{\ell-1},y,\chunk{x}{\ell+1}{L})-g(\chunk{x}{1}{L})\right|\\\nonumber&\phantom{=\int_Xset()}\left|\alpha_{\beta^{(\ell)}}(x^{(\ell)},y)-\alpha_{\beta'^{(\ell)}}(x^{(\ell)},y)\right|q_{\covmat^{(\ell)}}(y-x^{(\ell)})
\ud y\\
\nonumber
&\leq\int_\Xset \left(V_{\beta_0}(y)+V_{\beta_0}(x^{(\ell)})+2\sum_{k\neq\ell}V_{\beta_0}(x^{(k)})\right)\\\nonumber&\phantom{=\int_Xset()}\left|\alpha_{\beta^{(\ell)}}(x^{(\ell)},y)-\alpha_{\beta'^{(\ell)}}(x^{(\ell)},y)\right|q_{\covmat^{(\ell)}}(y-x^{(\ell)})
\ud y
\end{align*}
Applying \eqref{eq:bound-integral}, there exists $K<\infty$ such that for any $\ell\in\{1,\dots,L\}$,
we get
\begin{equation}
\label{eq:M-M-ell}
\left|\check{M}_{(\covmat^{(\ell)},\beta^{(\ell)})}g(\bx)
    -\check{M}_{(\covmat^{(\ell)},\beta'^{(\ell)})}g(\bx)\right|
\leq K|\beta^{(\ell)}-\beta'^{(\ell)}|\sum_{k=1}^LV_{\beta_0}(x^{(k)})
\end{equation}
Taking $\ell=1$ establishes \eqref{eq:prodM-M} with $k=1$.
Assume now that \eqref{eq:prodM-M} is satisfied  for some $k\in\{2,\dots, L-1\}$.
We have
\begin{align*}
\big| \bM^{(k+1)}_{(\bcovmat,\bbeta)}g(\bx)
     &-\bM^{(k+1)}_{(\bcovmat,\bbeta')}g(\bx)\big| \\
&\leq \big|\left(\check{M}_{(\covmat^{(k+1)},\beta^{(k+1)})}
-\check{M}_{(\covmat^{(k+1)},\beta'^{(k+1)})}\right)
\bM^{(k)}_{(\bcovmat,\bbeta)}g(\bx)\big|
\\
&\phantom{\leq} +\big|\check{M}_{(\covmat^{(k+1)},\beta'^{(k+1)})}
\big(\bM^{(k)}_{(\bcovmat,\bbeta)}-\bM^{(k)}_{(\bcovmat,\bbeta')}\big)
    g(\bx)\big|
\end{align*}
For any $\Vnorm{g}{\bV_{\beta_0}} \leq 1$, we have
\begin{equation*}
\big| \bM^{(k)}_{(\bcovmat,\bbeta)} g(\bx)\big| \leq \bM^{(k)}_{(\bcovmat,\bbeta)} \bV_{\beta_0} (\bx)
= \sum_{\ell=1}^k M_{(\covmat^{(\ell)},\beta^{(\ell)})} V_{\beta_0}(x^{(\ell)}) +
\sum_{\ell=k+1}^L V_{\beta_0}(x^{(\ell)}) \eqsp.
\end{equation*}
\autoref{lem:small-set-RWM} implies that,  for any 
$\Vnorm{g}{\bV_{\beta_0}} \leq 1$,
\[
\big| \bM^{(k)}_{(\bcovmat,\bbeta)} g(\bx)\big| \leq
\bV_{\beta_0}(\bx) + k b_{\epsilon,\beta_0} \eqsp,
\]
showing that $\|\bM^{(k)}_{(\bcovmat,\bbeta)}g\|_{\bV_{\beta_0}} < \infty$.
Hence by \eqref{eq:M-M-ell} there exists a constant $K^{(1)}_{k+1,\epsilon,\beta_0,\eta}<\infty$ such that
\begin{multline*}
\big|\big(\check{M}_{(\covmat^{(k+1)},\beta^{(k+1)})}
    -\check{M}_{(\covmat^{(k+1)},\beta'^{(k+1)})}\big)
    \bM^{(k)}_{(\bcovmat,\bbeta)}g(\bx)\big| \\
\leq K^{(1)}_{k+1,\epsilon,\beta_0,\eta}\; |\beta^{(k+1)}-\beta'^{(k+1)}|\bV_{\beta_0}(\bx) \eqsp.
\end{multline*}
By the induction assumption \eqref{eq:prodM-M} and \autoref{lem:small-set-RWM} there exists $K^{(2)}_{k+1,\epsilon,\beta_0,\eta}<\infty$ such that second term is bounded by
\begin{multline*} 
\big|M_{(\covmat^{(k+1)},\beta'^{(k+1)})}
    \big(\bM^{(k)}_{(\bcovmat,\bbeta)}
    -\bM^{(k)}_{(\bcovmat,\bbeta')}\big)g(\bx)\big|\\
\leq K_{k,\epsilon,\beta_0,\eta}|\bbeta-\bbeta'|M_{(\covmat^{(k+1)},\beta'^{(k+1)})}\bV_{\beta_0}(\bx)
\leq K^{(2)}_{k+1,\epsilon,\beta_0,\eta}|\bbeta-\bbeta'|\bV_{\beta_0}(\bx) \eqsp.
\end{multline*}
This shows that \eqref{eq:prodM-M} is satisfied for $k+1$. Carrying out the induction until $k=L$, there exists a constant $K_{L,\epsilon,\beta_0,\eta}<\infty$ such that,
for all $\Vnorm{g}{\bV_{\beta_0}} \leq 1$,
\begin{equation}\label{eq:prodM-M-pour-L}
R_1(x)= \left|\bM_{(\bcovmat,\bbeta)}g(\bx)-\bM_{(\bcovmat,\bbeta')}g(\bx)\right|
\leq K_{L,\epsilon,\beta_0,\eta}|\bbeta-\bbeta'|\bV_{\beta_0}(\bx) \eqsp,
\end{equation}
where $R_1$ is defined in \eqref{eq:definition-R}.

Consider now $R_2$. It is easy to see that by \eqref{eq:qcovmat} we
obtain analogous formula to \eqref{eq:M-M-ell}; that is, with the same
temperatures $\bbeta$ but different covariance matrices $\bcovmat$ and
$\bcovmat'$. The proof is concluded by using the same induction proof
as for the term $R_1$.
\end{proof}
\begin{lemma}\label{lem:continuity-S} 
Let $\beta_0 > 0$ and $\eta > 0$.  Then, there exists a constant
$K_{\beta_0,\eta}$ such that, for any $\bbeta,\bbeta' \in
\mathcal{K}_{\beta_0,\eta}$ and for any measurable 
function $g$ with $\Vnorm{
g}{\bV_{\beta_0}}\leq1$, it holds that \[ |\bS_{\bbeta}g(\bx)
-\bS_{\bbeta'}g(\bx)|\leq
K_{\beta_0,\eta}|\bbeta-\bbeta'|\bV_{\beta_0}(\bx)\eqsp. \]
\end{lemma} 
\begin{proof} 
Using the definition \eqref{eq:definitionSn} of $\bS_{\bbeta}$, we get
\begin{multline}
    \label{eq:S-V}
|\bS_{\bbeta}g(\bx) -\bS_{\bbeta'}g(\bx)|=\frac{1}{L-1}\Big|\sum_{\ell=1}^{L-1}\left(\varpi^{(\ell)}_{\bbeta} (x^{(\ell)},x^{(\ell+1)})-\varpi^{(\ell)}_{\bbeta'} (x^{(\ell)},x^{(\ell+1)})\right)\\
\times\left(g(\chunk{x}{1}{\ell-1},x^{(\ell+1)},x^{(\ell)},\dots,x^{(L)})-
g(\chunk{x}{1}{\ell-1},x^{(\ell)},x^{(\ell+1)},\dots,x^{(L)})\right)\Big|\eqsp.
\end{multline}
By \eqref{eq:probabilite-swap}, it holds that
$\varpi^{(\ell)}_{\bbeta} (x^{(\ell)},x^{(\ell+1)})= \varpi^{(\ell)}_{\bbeta'} (x^{(\ell)},x^{(\ell+1)})=1$ whenever $ \pi(x^{(\ell+1)}) \geq \pi(x^{(\ell)})$.
Therefore, using  \autoref{lem:sup-z-beta-beta}, we get that
\begin{align}
\label{eq:diff-swap-probabilities}
& \big|\varpi^{(\ell)}_{\bbeta} (x^{(\ell)},x^{(\ell+1)})
-\varpi^{(\ell)}_{\bbeta'} (x^{(\ell)},x^{(\ell+1)})\big|\\
\nonumber
&=\1_{\{\pi(x^{(\ell+1)})\leq\pi(x^{(\ell)}) \}} \left|\left(\frac{\pi(x^{(\ell+1)})}{\pi(x^{(\ell)})}\right)^{\beta^{(\ell)}-\beta^{(\ell+1)}}-\left(\frac{\pi(x^{(\ell+1)})}{\pi(x^{(\ell)})}\right)^{\beta'^{(\ell)}-\beta'^{(\ell+1)}}\right|\\
&\leq\frac{|\beta^{(\ell+1)}-\beta'^{(\ell+1)}|+|\beta^{(\ell)}-\beta'^{(\ell)}|}{\left(\beta^{(\ell)}-\beta^{(\ell+1)}\right)\wedge\left(\beta'^{(\ell)}-\beta'^{(\ell+1)}\right)}
\eqsp.
\nonumber
\end{align}
Since $\bbeta, \bbeta' \in \mathcal{K}_{\beta_0,\eta}$, 
$\max\{\left(\beta^{(\ell)}-\beta^{(\ell+1)}\right),
  \left(\beta'^{(\ell)}-\beta'^{(\ell+1)}
\right)\} \geq \eta$ for all $\ell\in\{1,\dots,L\}$.
Because $\Vnorm{g}{\bV_{\beta_0}}\leq1$ and $\bV_{\beta_0}$ are
invariant under permutations, we have by 
\eqref{eq:S-V} and \eqref{eq:diff-swap-probabilities}
\[
    |\bS_{\bbeta}g(\bx) -\bS_{\bbeta'}g(\bx)|\leq\frac{
4}{(L-1)\eta}\bV_{\beta_0}(\bx)\sum_{\ell=1}^{L}|\beta^{(\ell)}-\beta'^{(\ell)}|\eqsp.
\qedhere
\]
\end{proof} 
Now we are ready to conclude with the continuity of the parallel tempering
kernels.
\begin{proof}[Proof of \autoref{theo:continuity-P}]
The definition \eqref{eq:lyapunov-function-junct} of $\bV_{\bbeta_0}$ implies that,
for any $\alpha \in \ocint{0,1}$,
\[
L^{\alpha-1} \sum_{i=1}^L V_{\alpha \beta_0}(x^{(i)}) \leq
\left( \sum_{i=1}^L V_{\beta_0}(x^{(i)}) \right)^\alpha \leq \sum_{i=1}^L V_{\alpha \beta_0}(x^{(i)})
\eqsp,
\]
showing that $L^{\alpha-1} \bV_{\alpha \beta_0}(\bx) \leq \bV^\alpha_{\beta_0}(\bx) \leq  \bV_{\alpha \beta_0}(\bx)$. Therefore, the norms $\Vnorm{\cdot}{\bV^\alpha_{\beta_0}}$ and $\Vnorm{\cdot}{\bV_{\alpha \beta_0}}$ are
equivalent. It suffices to prove the results with $\alpha=1$. Write
\[
|\bP_{(\bcovmat,\bbeta)}g(\bx) -\bP_{(\bcovmat',\bbeta')}g(\bx)|\leq T_1(\bx)+T_2(\bx)
\eqsp,
\]
where
\begin{align*}
T_1(\bx)&\eqdef\left|\bS_{\bbeta}(\bM_{(\bcovmat,\bbeta)}-\bM_{(\bcovmat',\bbeta')})g(\bx)\right|\\
T_2(\bx)&\eqdef\left|(\bS_{\bbeta}-\bS_{\bbeta'})\bM_{(\bcovmat',\bbeta')}g(\bx)\right|\eqsp.
\end{align*}
By \autoref{prop:continuity-bM}, we obtain
\begin{multline*}
T_1(\bx)\leq K_{\epsilon,\beta_0,\eta}
\left\{|\bbeta-\bbeta'|+|\bcovmat-\bcovmat'|\right\}\bS_{\bbeta}\bV_{\beta_0}(\bx)
\\
\leq K_{\epsilon,\beta_0,\eta} \, \left\{|\bbeta-\bbeta'|+|\bcovmat-\bcovmat'|\right\}\bV_{\beta_0}(\bx)\eqsp.
\end{multline*}
By \eqref{eq:la-preuve-pour-bM} of
\autoref{lem:small-set-RWM}, 
we obtain that $\Vnorm{\bM_{(\bcovmat',\bbeta')}g}{\bV_{\beta_0}}\leq C$. Hence, by \autoref{lem:continuity-S} we get that
\[T_2(\bx)\leq C K_{\beta_0,\eta} |\bbeta-\bbeta'|\bV_{\beta_0}(\bx)\eqsp,\]
which concludes the proof.
\end{proof}


\subsection{Proof of \autoref{theo:law-of-large-numbers}} 
We now turn into the proof of the strong law of large numbers. We
start by gathering some known results and by technical lemmas.

\begin{lemma} 
\label{lem:expectation_V} Assume \A{a:super-exp} and that, in addition,  $\E \left[ \bV_{\beta_0}(\bX_0) \right]<\infty$. Then,
\[
\sup_{n \geq 1} \E \left[ \bV_{\beta_0}(\bX_n) \right]<\infty \eqsp,
\]
where $\bX_n$ is the state of the adaptive parallel tempering algorithm defined in \eqref{eq:definition-state-pt}.
\end{lemma} 
\begin{proof} 
    Under \A{a:super-exp}, by \autoref{prop:drift}, for all $n\in\nset$ we have that
\begin{equation}
\label{eq:drift_n}
\bP_{(\bcovmat_n,\bbeta_n)} \bV_{\beta_0} \leq \lambda_{\epsilon,\beta_0} \bV_{\beta_0}+ b\eqsp.
\end{equation}
Iterating \eqref{eq:drift_n}, by \eqref{def:expn} we obtain
\begin{multline*}
\E[\bV_{\beta_0}(\bX_n)]=\E\left[\cesp{\bV_{\beta_0}(\bX_n)}{\F_{n-1}}\right]= \E \left[ \bP_{(\bcovmat_{n-1},\bbeta_{n-1})} \bV_{\beta_0}(\bX_{n-1}) \right] \\
\leq\lambda_{\epsilon,\beta_0}\E\left[\bV_{\beta_0}(\bX_{n-1})\right]+b_{\epsilon,\beta_0} \eqsp.
\end{multline*}
By iterating this majorisation, we get recursively
\begin{align*}
\E[\bV_{\beta_0}(\bX_n)] 
&\leq\lambda_{\epsilon,\beta_0}^n\E\left[\bV_{\beta_0}(\bX_{0})\right]+b\sum_{k=1}^n\lambda_{\epsilon,\beta_0}^k
\\ &\leq\lambda_{\epsilon,\beta_0}\E\left[\bV_{\beta_0}(\bX_{0})\right]+\frac{b}{1-\lambda_{\epsilon,\beta_0}}\eqsp.
\end{align*}
Because the term on the right is independent from $n$ and since, by assumption, $\E \bV_{\beta_0}(\bX_0)<\infty$, the proof is concluded.
\end{proof} 

The following Lemma is adapted from \cite[Lemma~4.2]{fort-moulines-priouret}.
\begin{lemma}\label{lem:continuity-Poisson} 
  Assume \A{a:super-exp}. Let $\epsilon > 0$, $\beta_0 > 0$, $\eta > 0$.
  For any $(\bcovmat,\bbeta) \in \mathcal{K}_{\epsilon,\beta_0,\eta} \eqdef \positivematrixset[\epsilon]{d}[L] \times \mathcal{K}_{\beta_0,\eta}$, let
  $F_{(\bcovmat,\bbeta)}: \Xset^L \to \rset^+$ be a measurable function such that
  \begin{equation}
  \label{eq:bound-F}
  \sup_{(\bcovmat,\bbeta) \in \mathcal{K}_{\epsilon,\beta_0,\eta} } \Vnorm{F_{(\bcovmat,\bbeta)}}{\bV_{\beta_0}} < +\infty \eqsp.
  \end{equation}
  Define $\hat{F}_{(\bcovmat,\bbeta)}$ the solution of the Poisson equation (see \autoref{coro:Poisson})
\[
\hat{F}_{(\bcovmat,\bbeta)} \eqdef \sum_{n \geq 0} \bP_{(\bcovmat,\bbeta)}^n\{F_{(\bcovmat,\bbeta)} - \bpi_{\bbeta}(F_{(\bcovmat,\bbeta)}) \} \eqsp.
\]
There exist a constant $L_{\epsilon,\beta_0} < \infty$ such that, for any ${(\bcovmat,\bbeta)}, {(\bcovmat',\bbeta')} \in \mathcal{K}_{\epsilon,\beta_0,\eta}$,
\begin{equation}
\label{eq:differencepi}
\Vnorm{\bpi_{\bbeta} - \bpi_{\bbeta'}}{\bV_{\beta_0}} \\
\leq L_{\epsilon,\beta_0}  \ D_{\bV_{\beta_0}}\left[(\bcovmat,\bbeta),(\bcovmat',\bbeta')\right] \eqsp,
\end{equation}
and
\begin{align}
\label{eq:differencePF}
\big\|&\bP_{(\bcovmat,\bbeta)} \hat{F}_{(\bcovmat,\bbeta)} 
  - \bP_{{(\bcovmat',\bbeta')}} \hat{F}_{{(\bcovmat',\bbeta')}}\|_{\bV_{\beta_0}}
\\
&\leq L_{\epsilon,\beta_0} \bigg\{ 
  \Vnorm{F_{(\bcovmat,\bbeta)} - F_{{(\bcovmat',\bbeta')}}}{\bV_{\beta_0}}
\nonumber\\
&\phantom{ \leq L_{\epsilon,\beta_0} \bigg\{ }
  +\sup_{(\bcovmat,\bbeta) \in \mathcal{K}_{\epsilon,\beta_0,\eta}}
\Vnorm{F_{(\bcovmat,\bbeta)}}{\bV_{\beta_0}} \
D_{\bV_{\beta_0}}\left[(\bcovmat,\bbeta),(\bcovmat',\bbeta') \right]
\bigg\} \eqsp.
\nonumber
\end{align}
\end{lemma} 

We use repeatedly the following elementary result on 
a projection to a closed convex set.
\begin{lemma}\label{lem:proj-distance} 
Let $\set{E}$ be an Euclidean space. 
Given any nonempty closed convex set $\set{K} \subset \set{E}$, 
denote by $\Pi_{\set{K}}$ the projection on the set $\set{K}$. 
For any $(x,y) \in \set{E} \times \set{E}$, 
$ \Vert \Pi_{\set{K}}(x) - \Pi_{\set{K}}(y) \Vert \leq \Vert x - y \Vert$, 
where $\Vert \cdot \Vert$ is the Euclidean norm.
\end{lemma} 

\begin{lemma}\label{lem:DV-n-n+1} 
    Assume \A{a:super-exp} and $\sup_{n\in\nset}\gamma_{n,2}\leq 1$, where $\gamma_{n,2}$ is defined in \eqref{eq:am-cov} and \eqref{eq:am-mean}. Then, for all $\kappa>0$ and $\alpha \in \ccint{0,1}$,
there exists constant $K_{\kappa,\alpha,\epsilon,\beta_0,\eta}<\infty$ such that for any $n\in\nset$
it holds
\begin{multline*}
D_{\bV^\alpha_{\bbeta_0}}\left[ (\bcovmat_n,\bbeta_n), (\bcovmat_{n+1},\bbeta_{n+1}) \right]\\
\leq K_{\kappa,\epsilon,\beta_0,\eta} \gamma_{n+1} \left[\bV^\kappa_{\beta_0}(\bX_{n+1})+\sum_{k=0}^n a_{n,k} \bV^\kappa_{\beta_0}(\bX_{k})\right]\eqsp,
\end{multline*}
where $\gamma_n = \sum_{i=1}^3 \gamma_{n,i}$, and 
$a_{n,n} \eqdef \gamma_{n,2}$ and for $k \in \{0, \dots, n-1\}$
\begin{equation}
\label{eq:definition-a}
a_{n,k} \eqdef \gamma_{k,2} \prod_{j=k+1}^n (1-\gamma_{j,2}) \eqsp.
\end{equation}
\end{lemma} 
\begin{proof} 
According to \autoref{theo:continuity-P}, under \A{a:super-exp},
there exists a constant $K_{\epsilon,\alpha,\beta_0,\eta}$ such that
\[
D_{\bV_{\bbeta_0}^\alpha}\left[ (\bcovmat_n,\bbeta_n), (\bcovmat_{n-1},\bbeta_{n-1}) \right]
\leq K_{\epsilon,\alpha,\beta_0,\eta} \left\{|\bbeta_n-\bbeta_{n+1}|+|\bcovmat_n-\bcovmat_{n+1}|\right\} \eqsp.
\]
For any $\ell\in\{1,\dots,L-1\}$ consider 
$|\rho^{(\ell)}_{n+1}-\rho^{(\ell)}_{n}|$, where
$\seq{\smash{\rho^{(\ell)}}}$ 
is defined by \eqref{def:temp_adapt}.
Since, by \eqref{eq:definition-H-ell}, 
$|H^{(\ell)}(\chunk{\rho}{1}{\ell},x)|\leq 1$ for any $\chunk{\rho}{1}{\ell} \in \rset^\ell$ and $x \in \Xset$, applying \autoref{lem:proj-distance} we obtain
\begin{equation}\label{eq:diff-rho-ell}
|\rho^{(\ell)}_{n+1}-\rho^{(\ell)}_{n}|\leq \gamma_{n,1} |H^{(\ell)}(\chunk{\rho_{n-1}}{1}{\ell},\bX_n)|\leq \gamma_{n,1}\eqsp.
\end{equation}
Define the function 
\[
\bbeta: \chunk{\rho}{1}{L-1} \to (1, \beta^{(2)}(\rho^{(1)}), \dots, \beta^{(L)}(\chunk{\rho}{1}{L-1}))
\]
where the functions $\beta^{(\ell)}$ are defined in \eqref{eq:beta-rho-def}.
The function $\bbeta$ is continuously differentiable. By definition \eqref{def:temp_adapt}, for all $n\in\nset$,  $\brho_n\in\ccint{\underline{\rho},\overline{\rho}}^{L-1}$. Hence \eqref{eq:diff-rho-ell} implies that there exists $K<\infty$ such that
\begin{equation}\label{eq:diff-beta}
|\bbeta_n - \bbeta_{n+1}| = |\bbeta(\brho_n)-\bbeta(\brho_{n+1})|\leq K\gamma_{n,1}\eqsp.
\end{equation}
Now consider $|\covmat_{n}-\covmat_{n+1}|$. By definition \eqref{eq:definition-covmat_n} we get
\[
|\covmat^{(\ell)}_{n}-\covmat^{(\ell)}_{n+1}|\leq \exp(T^{(\ell)}_n)|\Gamma^{(\ell)}_n-\Gamma^{(\ell)}_{n+1}|+|\exp(T^{(\ell)}_n)-\exp(T^{(\ell)}_{n+1})||\Gamma^{(\ell)}_{n+1}|\eqsp.
\]
The scale adaptation procedure \eqref{eq:scaling-adaptation} by \autoref{lem:proj-distance} satisfies
\[|T^{(\ell)}_n-T^{(\ell)}_{n+1}|\leq \gamma_{n+1,3}\eqsp,\]
which implies that there exists constant $K<\infty$ such that $|\exp(T^{(\ell)}_n)-\exp(T^{(\ell)}_{n+1})|\leq K\gamma_{n+1,3}$.  By \eqref{eq:am-cov} and \autoref{lem:proj-distance} we obtain
\begin{align*}
|\Gamma^{(\ell)}_n-\Gamma^{(\ell)}_{n+1}|&\leq \gamma_{n+1,2}\left|(X_{n+1}^{(\ell)}-\mu_{n}^{(\ell)})(X_{n+1}^{(\ell)}-\mu_{n}^{(\ell)})^T
+\Gamma^{(\ell)}_n
\right|
\\
&\leq \gamma_{n+1,2}\left[2| X^{(\ell)}_{n+1}|^2+ 2|\mu^{(\ell)}_n|^2 + |\Gamma_n^{(\ell)}| \right]
\eqsp.
\end{align*}
By definition \eqref{eq:am-cov} and \eqref{eq:scaling-adaptation} $\exp(T_n)$ and $|\Gamma_n|$ are uniformly bounded for all $n\in\nset$.
Hence, gathering all terms, there exists a constant $K < \infty$ such that
\[D_{\bV_{\bbeta_0}^\alpha}\left[ (\bcovmat_n,\bbeta_n), (\bcovmat_{n-1},\bbeta_{n-1}) \right]
\leq K \gamma_{n+1} \left[| X^{(\ell)}_{n+1}|^2+ |\mu^{(\ell)}_n|^2\right]\eqsp,
\]
where $\gamma_n =\sum_{i=1}^3 \gamma_{n,i}$. Under assumption
$\sup_{n\in\nset}\gamma_{n,2} \leq 1$, by \eqref{eq:am-mean}, we get
that for any $n\in \nset$ $\mu^{(\ell)}_n=\sum_{k=0}^n a_{n,k}
X_{n-k}^{(\ell)}$, where the positive weights $a_{n,k}$, $k \in
\{0,\dots,n\}$ are defined in \eqref{eq:definition-a}. 
Because $\sum_{k=0}^n a_{n,k}=1$, the Jensen
inequality implies 
\[
    |\mu^{(\ell)}_n|^2\leq\left(\sum_{k=0}^n a_{n,k}
|X_{n-k}^{(\ell)}|\right)^2\leq\sum_{k=0}^n a_{n,k}
|X_{n-k}^{(\ell)}|^2\eqsp.
\] Finally, under \A{a:super-exp} for any
$\kappa>0$ there exists $K_\kappa$ such that, for all $\bx \in
\Xset^L$, $| \bx |^2 \leq K_\kappa\bV_{\beta_0}^\kappa(\bx)$.
\end{proof}

\begin{lemma}\label{lem:series-DV} 
    Assume \A{a:super-exp} and $\sup_{n\in\nset}\gamma_{n,2}\leq 1$. 
For any non-negative sequence $\sequence{b}[n]$ satisfying $\sum_{n\geq1}b_n(\gamma_{n,1}+\gamma_{n,2}+\gamma_{n,3})<\infty$ and for all $\alpha\in\ooint{0,1}$, we have
\[\sum_{n=1}^\infty b_n D_{\bV_{\bbeta_0}^\alpha}\left[ (\bcovmat_n,\bbeta_n), (\bcovmat_{n-1},\bbeta_{n-1}) \right] \ \bV_{\beta_0}^\alpha(\bX_n) < +\infty\ \as\eqsp.\]
\end{lemma} 
\begin{proof} 
Since $\sup_{n\in\nset}\gamma_{n,2}\leq 1$, under \A{a:super-exp} by \autoref{lem:DV-n-n+1} for all $\kappa>0$ there exists  $K<\infty$ such that
\begin{multline*}
D_{\bV^\alpha_{\bbeta_0}}\left[ (\bcovmat_n,\bbeta_n), (\bcovmat_{n-1},\bbeta_{n-1}) \right]\\
\leq K \gamma_{n+1} \left[\bV^\kappa_{\beta_0}(\bX_{n+1})+\sum_{k=0}^n a_{n,k}\bV^\kappa_{\beta_0}(\bX_{k})\right]\eqsp,
\end{multline*}
where $\gamma_{n} = \sum_{i=1}^3 \gamma_{n,i}$ and the triangular
array $\{a_{n,k} \}_{k=0}^n$ is defined in \eqref{eq:definition-a}.
Set $\kappa=1-\alpha$. Since $\sum_{n \geq 1} b_n \gamma_n < \infty$, 
it is enough to show that
\[
A \eqdef \sup_{n\in\nset}\E\left\{\left(\bV^{1-\alpha}_{\beta_0}(\bX_{n+1})+\sum_{k=0}^n a_{n,k} \bV^{1-\alpha}_{\beta_0}(\bX_{k})\right)\bV_{\beta_0}^\alpha(\bX_{n+1})\right\} < \infty
\eqsp.
\]
By H{\"o}lder inequality we get for all $k,n\in\nset$
\[
\E\left\{ \bV_{\beta_0}^{1-\alpha}(\bX_{k})\bV_{\beta_0}^\alpha(\bX_{n+1})\right\}
\leq \left\{\E\bV_{\beta_0}(\bX_{n+1})\right\}^\alpha \, \left\{\E\bV_{\beta_0}(\bX_{k})\right\}^{1-\alpha}
\eqsp.
\]
Since the weights $a_{n,k}$ are non-negative and $\sum_{k=0}^n
a_{n,k}=1$, we get $A \leq2\sup_{n\in\nset}\E\bV_{\beta_0}(\bX_{n})$
the proof is concluded by applying \autoref{lem:expectation_V}.
\end{proof} 

\begin{proof}[Proof of \autoref{theo:law-of-large-numbers}] 
\blazej{need to check which corollary we are speaking about} According
to \cite[Corollary~2.9]{fort-moulines-priouret}, we need to check that
\textbf{(i)} $\sup_{n \geq 0} \E[\bV(\bX_n)] < +\infty$, and
\textbf{(ii)} there exists $\alpha \in \ooint{0,1}$ such that
\[
    \sum_{n=1}^\infty n^{-1} D_{\bV_{\bbeta_0}^\alpha}\left[
(\bcovmat_n,\bbeta_n), (\bcovmat_{n-1},\bbeta_{n-1}) \right] \
\bV_{\beta_0}^\alpha(\bX_n) < +\infty\ \as,
\]
which follow from 
\autoref{lem:expectation_V} and \autoref{lem:series-DV}, respectively.
\end{proof} 


\subsection{Proof of \autoref{prop:uniqueness-temperature}} 

In order to prove the existence and uniqueness of the root
of the mean field $\bh$, we first introduce some notation and prove 
the key \autoref{lem:monotonic:h}.

Below, we omit the set $\Xset$ from the integrals.
Let us define a symmetric function $\tilde{h}:(0,1]^2\to [0,1]$
as follows
\begin{equation}
\label{eq:definition-tilde-h}
\tilde{h}(u,v)=
\iint \bigg(1\wedge\frac{\pi^{u}(x)\pi^{v}(y)}{\pi^{u}(y)\pi^{v}(x)}\bigg)
\frac{\pi^v(\ud x)}{Z(v)} \frac{\pi^u(\ud y)}{Z(u)} \eqsp.
\end{equation}
Note that for all $(u,v) \in \ocint{0,1}^2$, we get
\begin{equation}
\label{eq:definition-tildeh}
\tilde{h}(u,v)= \int f_v( \pi(y) ) \frac{\pi^u(y)}{Z(u)} \ud y = 
\int f_u( \pi(y) ) \frac{\pi^v(y)}{Z(v)} \ud y
\eqsp,
\end{equation}
where for $v \in \ocint{0,1}$, 
\begin{equation}
\label{eq:definition-fv}
f_v(z)\eqdef \int \frac{\pi^v(\ud x)}{Z(v)} \left[\1_{\{ \pi\leq z\}}(x)+ \1_{\{\pi<z\}}(x) \right] \eqsp.
\end{equation}

\begin{lemma}\label{lem:monotonic:h} 
Assume that $\lleb(\Xset) = \infty$ and the density $\pi$ is positive, bounded,
and $\int \pi^\kappa(x) \lleb(\ud x)<\infty$ for all $0<\kappa\le 1$.
Then, for all
$v\in(0,1]$ the function $u\mapsto \tilde{h}(u,v)$ restricted to $(0,v)$
is differentiable
and monotonic with the following derivative
\begin{multline*}
\frac{\partial \tilde{h}}{\partial u}(u,v)
    =\frac{1}{2}\iint \left[f_v(\pi(y))-f_v(\pi(z)) \right] \\ \times
    \left[\log(\pi(y))-\log(\pi(z)) \right]
      \frac{\pi^u(\ud y)}{Z(u)}\frac{\pi^u(\ud z)}{Z(u)} \eqsp.
\end{multline*}
Moreover, $\lim_{u\to 0}\tilde{h}(u,v)=0$ and $\lim_{u\to v}\tilde{h}(u,v)=1$.
\end{lemma} 
\begin{proof} 
We will first show that, for any bounded measurable function $f$ and $u\in(0,1)$,
\begin{equation}
    \frac{\ud}{\ud u}\left(
    \int_{\Xset} f(y)\pi^{u}(\ud y) \right)
    =\int_{\Xset} f(y)\log(\pi(y))\pi^{u}(\ud y)  \eqsp.
    \label{eq:diff-condition}
\end{equation}
    Since $\pi$ is bounded, there exists a function
    $J:\Xset\to[0,\infty)$ such that
    $\pi(x) = c \exp(-J(x))$.
    By the dominated convergence theorem, we only need to show that
    $|h^{-1}(\pi^{u+h}(y)-\pi^{u}(y)|$
is bounded uniformly for $h$ in some neighbourhood of $0$ by
an integrable function depending only on $y$. We may write
\[
    \left|\frac{\pi^{u+h}(y)-\pi^{u}(y)}{h}\right|
    =c\exp(-u J(y))\frac{|\exp(-h J(y))-1|}{|h|}\eqsp .
\]
Applying the mean value theorem we obtain
$|\exp(-h J(y))-1|\leq |h|\exp(|h|J(y))$.
Hence for all $|h|\leq\frac{u}{2}$ we get that
\[
    \left|\frac{\pi^{u+h}(y)-\pi^{u}(y)}{h}\right|
        \leq c \exp\bigg(-\frac{u J(y)}{2}\bigg)
        = c \pi^{u/2}(y)\eqsp ,
\]
which concludes the proof of \eqref{eq:diff-condition}.

For any given $v$, using \eqref{eq:diff-condition}, we compute the partial derivative of $\tilde{h}$ with respect to $u$
\begin{align*}
&\frac{\partial \tilde{h}}{\partial u}(u,v) \\
&=\int f_v(\pi(y))\log(\pi(y))\frac{\pi^u(\ud y)}{Z(u)}
 -\int f_v(\pi(y))\frac{\pi^u(\ud y)}{Z(u)}
  \int\log(\pi(y))\frac{\pi^u(\ud y)}{Z(u)} \\
&=\frac{1}{2}\int\left[f_v(\pi(y))-f_v(\pi(z)) \right]
   \left[\log(\pi(y))-\log(\pi(z)) \right]
       \frac{\pi^u(\ud y)}{Z(u)}\frac{\pi^u(\ud z)}{Z(u)} \eqsp.
\end{align*}
Since the functions $z \mapsto f_v(z)$ and $z \mapsto \log(z)$ are 
non-decreasing, 
$[f_v(\pi(y))-f_v(\pi(z)) ][\log(\pi(y))-\log(\pi(z))]\geq 0$
for all $y,z$.
Moreover, because $\pi$ is positive the $(\pi\times\pi)$-measure of
the sets $\{y,z\in\Xset\eqsp:\eqsp\pi(y)\neq\pi(z)\}$ must be
positive due to $\lleb(\Xset)=\infty$.
Hence $\frac{\partial \tilde{h}}{\partial u}(u,v)>0$
for $u\in(0,v)$ and this completes the proof of the first part.

Since $\pi^{u}\vee\pi^{v}$ is integrable, the dominated convergence theorem
implies $\tilde{h}(u,v)\to 1$ as $u\to v$.
For the second limit consider \eqref{eq:definition-tildeh}.
For any
$\varepsilon>0$ we can find
a $\delta>0$ such that $f_v(\delta)\leq\varepsilon/2$.
Therefore,
\[
    \int f_v(\pi(y))\frac{\pi^u(\ud y)}{Z(u)}
    \leq \frac{\varepsilon}{2} \int \frac{\pi^u(\ud y)}{Z(u)}
    +\frac{\int \1_{\{\pi> \delta\}}(y)\pi^{u}(\ud y)}{\int\pi^{u}(\ud y)}.
\]
There exists a constant $c<\infty$ such that
$\int \1_{\{ \pi> \delta\}} (y) \pi^{u}(\ud y)\leq c$
for all $u\in[0,1]$
and
\[
    \int\pi^{u}(\ud y) \geq \int\pi^{u}(\ud y) \1_{\{\pi\leq1\}}(y) \eqsp.
\]
Observe that, for any $y \in \Xset$, the function  $u \mapsto\pi^{u}(y) \1_{\{\pi\leq1\}}(y)$ is
non-increasing. Therefore, using the monotone convergence theorem
and $\lleb(\Xset)= \infty$, we get
$\lim_{u\to 0}\int\pi^{u}(\ud y)=\infty$.
Hence we can find $u_0$ such that for all $u<u_0$ the
normalising constant $\int\pi^{u}( \ud y) \ge c/\epsilon$ and therefore
$\int f_v(\pi(y))\frac{\pi^u(\ud y)}{Z(u)} \le \epsilon$.
\end{proof} 

\begin{proof}[Proof of \autoref{prop:uniqueness-temperature}] 
    We recall that $h^{(1)}(\rho)$ depends only on $\rho^{(1)}$ and we
    may write
 $h^{(1)}(\rho)=\tilde{h}(\exp(-\exp(\rho^{(1)})),1)-\alpha^*$, where $\tilde{h}$ is defined in \eqref{eq:definition-tilde-h}.
By \autoref{lem:monotonic:h} the function $u \mapsto \tilde{h}(u,1)$ is strictly monotonic on $\ccint{0,1}$ and $\lim_{u \to 0} \tilde{h}(u,1)=0$ and $\lim_{u \to 1} \tilde{h}(u,1)=1$. Since $u \to \tilde{h}(u,1)$ is continuous, there exists a unique
$\hat{\rho}^{(1)}$ such that
$h^{(1)}(\hat{\rho}^{(1)},\rho^{(2)},\ldots,\rho^{(L-1)})
=h^{(1)}(\hat{\rho}^{(1)})=0$.
We proceed by induction and
assume the existence of unique
$\hat{\rho}^{(1)},\dots,\hat{\rho}^{(\ell-1)}$ such that for any $k<\ell$
we have $h^{(k)}(\hat{\rho}^{(1)},\dots,\hat{\rho}^{(k)})=0$.
Denoting
\begin{align}
\label{eq:definition-u}
&u(\rho)= \exp( - \exp( \rho) ) \\
\label{eq:definition-v}
&v_{\ell-1}(\chunk{\rho}{1}{\ell-1})=\exp\big(-\sum_{i=1}^{\ell-1}\exp(\rho^{(i)})\big) \eqsp,
\end{align}
we may write
\[
h^{(\ell)}(\chunk{\hat{\rho}}{1}{\ell-1},\rho^{(\ell)})
=\tilde{h}(u(\rho^{(\ell)})v_{\ell-1}(\chunk{\hat{\rho}}{1}{\ell-1}),v_{\ell-1}(\chunk{\hat{\rho}}{1}{\ell-1})) - \alpha^* \eqsp.
\]
We conclude as before from the monotonicity
of $u \mapsto \tilde{h}(u,v)$ and the limits $\lim_{u \to 0} \tilde{h}(u,v)=0$ and
$\lim_{u \to v} \tilde{h}(u,v)=1$ (see \autoref{lem:monotonic:h}).
\end{proof}


\subsection{Proof of \autoref{theo:convergence-temp-adapt}} 
Our proof of convergence of the temperature adaptation employs
classical convergence results on stochastic approximation. It is,
however, not easy to construct a `global' Lyapunov function for the
mean field $\bh$, but the specific structure of the problem allows to
deduce the convergence recursively for 
$\rho^{(1)},\ldots,\rho^{(L-1)}$.

In this section, for notational simplicity,
$\bP_{(\bcovmat,\bbeta(\brho))}= \bP_{(\bcovmat,\brho)}$ and 
$\bpi_{\bbeta(\brho)}= \bpi_{\brho}$.
For any $\ell=1,\dots,L-1$ we rewrite \eqref{def:temp_adapt} as follows
\begin{equation}\label{def:temp_adapt_second}
\rho^{(\ell)}_n = \Pi_{\rho} [ \rho^{(\ell)}_{n-1}+\gamma_{n,1} g^{(\ell)}(\rho_{n-1}^{(\ell)})
+\gamma_{n,1} \varepsilon^{(\ell)}_n+ \gamma_{n,1} r^{(\ell)}_n] \eqsp,
\end{equation}
where $\Pi_\rho$ is defined in \eqref{def:temp_adapt} and
\begin{align}
&g^{(\ell)}(\rho)= \tilde{h}\left(u(\rho)v_{\ell-1}(\chunk{\hat{\rho}}{1}{\ell-1}), v_{\ell-1}(\chunk{\hat{\rho}}{1}{\ell-1})\right) -\alpha^\star
\label{eq:definition-g}
\\
\label{eq:definition-epsilon}
&\varepsilon_n^{(\ell)}= H^{(\ell)}_{\chunk{\rho_{n-1}}{1}{\ell}} (\bX_{n})-
\pi_{\brho_{n-1}}\big[ 
H^{(\ell)}_{\chunk{\rho_{n-1}}{1}{\ell}}(\bX_{n}) \big] \\
\label{eq:definition-r}
&r_n^{(\ell)}=\tilde{h}
\big(u(\rho_{n-1}^{(\ell)})v_{\ell-1}(\chunk{\rho_n}{1}{\ell-1}),v_{\ell-1}(\chunk{\rho_n}{1}{\ell-1})\big)\\\nonumber
&\quad\quad\quad
-\tilde{h}\big(u(\rho_{n-1}^{(\ell)})v_{\ell-1}(\chunk{\hat{\rho}}{1}{\ell-1}),v_{\ell-1}(\chunk{\hat{\rho}}{1}{\ell-1})\big)
\eqsp,
\end{align}
where $\tilde{h}$ is defined in \eqref{eq:definition-tilde-h}, 
$H^{(\ell)}_{\chunk{\rho}{1}{\ell-1}}(\bx)$ is a shorthand notation for
$H^{(\ell)}(\chunk{\rho}{1}{\ell-1},\bx)$ defined in \eqref{eq:definition-H-ell},
$u$ and $v_{\ell-1}$ are defined in \eqref{eq:definition-u} and
\eqref{eq:definition-v}, respectively, 
and, by convention, $v_0=1$. 

We decompose the term
$\varepsilon^{(\ell)}_n$ as the sum of a martingale difference term and a remainder
term that goes to zero. To do this, we use the Poisson decomposition.
For $(\bcovmat,\brho) \in \positivematrixset[\varepsilon]{d}[L] \times
\ccint{\underline{\rho},\overline{\rho}}^{L-1}$ defined in 
Section \ref{sec:alg}, denote by $\hat{H}^{(\ell)}_{(\bcovmat,\brho)}$ the solution of the Poisson equation
\[
\hat{H}^{(\ell)}_{(\bcovmat,\brho)}(\bx)-\bP_{(\bcovmat,\brho)}
    \hat{H}^{(\ell)}_{(\bcovmat,\brho)}(\bx)=
    H^{(\ell)}_{\chunk{\rho}{1}{\ell}}(\bx)-\bpi_{\rho}[H^{(\ell)}_{\chunk{\rho}{1}{\ell}}]
    \eqsp ,
\]
which exists by \autoref{coro:Poisson}.
Hence, $\varepsilon_n^{(\ell)} = \delta M^{(\ell)}_n + \kappa_n^{(\ell)}$
where
\begin{align*}
& \delta M^{(\ell)}_n \eqdef \hat{H}^{(\ell)}_{\bcovmat_{n-1},\brho_{n-1}}(\bX_n) - \bP_{(\bcovmat_{n-1},\brho_{n-1})} \hat{H}^{(\ell)}_{\bcovmat_{n-1},\brho_{n-1}}(\bX_{n-1}) \\
& \kappa_n^{(\ell)} \eqdef  \bP_{(\bcovmat_{n-1},\brho_{n-1})} \hat{H}^{(\ell)}_{\bcovmat_{n-1},\brho_{n-1}}(\bX_{n-1}) -
\bP_{(\bcovmat_{n-1},\brho_{n-1})} \hat{H}^{(\ell)}_{\bcovmat_{n-1},\brho_{n-1}}(\bX_{n})\eqsp.
\end{align*}

\begin{lemma} \label{prop:noise} 
Assume \A{a:super-exp}-\A{a:step-size-v2} and $\E\bV_{\beta_0}(\bX_0)<\infty$.
For all $\ell\in\{1,\dots,L-1\}$ and $T < \infty$, it holds
\begin{align}\label{eq:martingale-noise}
&\lim_{n\to\infty}\sup_{n \leq k \leq m(n,T)}\left|\sum_{i=n}^k \gamma_{i,1} \delta M_i^{(\ell)}\right|=0\quad\as\eqsp,  \\
& \lim_{n\to\infty}\sup_{n \leq k \leq m(n,T)}\left|\sum_{i=n}^k \gamma_{i,1} \kappa_i^{(\ell)}\right|=0\quad\as \eqsp,
\label{eq:reminder}
\end{align}
where
\[
m(n,T)\eqdef\max\left\{j>n:\sum_{i=n+1}^j\gamma_{i,1}\leq T\right\} \eqsp.
\]
\end{lemma} 
\begin{proof} 
Consider \eqref{eq:martingale-noise}. Since $\delta M_i^{(\ell)}$ are
martingale increments, Doob's inequality implies 
\begin{multline}\label{eq:doob}
\E \left(\sup_{k\ge n}\left|\sum_{i=n}^k \gamma_{i,1}, \delta M_i^{(\ell)}\right| \right)^2
\\\le K \sum_{i=n}^\infty\gamma_{i,1}^2
\E\left[\left|
\hat{H}^{(\ell)}_{\bcovmat_{i-1},\brho_{i-1}}(\bX_{i})-\bP_{(\bcovmat_{i-1},\brho_{i-1})} \hat{H}^{(\ell)}_{\bcovmat_{i-1},\brho_{i-1}}(\bX_{i-1})\right|^2
\right]\eqsp .
\end{multline}
Since $|H^{(\ell)}|\leq 1$, \autoref{coro:Poisson} yields
\begin{align}\label{eq:inequality-poisson}
&\sup_{(\bcovmat,\brho)\in\positivematrixset[\epsilon]{d}[L]\times \ccint{\underline{\rho},\overline{\rho}}^{L-1}}
|\hat{H}^{(\ell)}_{(\bcovmat,\brho)}(\bx)| \leq K\bV^{1/2}_{\beta_0}(\bx) \\
&\sup_{(\bcovmat,\brho)\in\positivematrixset[\epsilon]{d}[L]\times \ccint{\underline{\rho},\overline{\rho}}^{L-1}}
|\bP_{(\bcovmat,\brho)}\hat{H}^{(\ell)}_{(\bcovmat,\brho)}(\bx)|\leq 
K\bV^{1/2}_{\beta_0}(\bx)\eqsp.
\label{eq:inequality-poisson2}
\end{align}
Hence by \eqref{eq:doob} we obtain
\[
\E \left( \sup_{k\ge n}\left|\sum_{i=n}^k \gamma_{i,1} \delta M_i^{(\ell)}\right|\right)^2\leq K\sum_{i=n}^\infty \gamma_{i,1}^2\E\bV_{\beta_0}(\bX_i)\eqsp.
\]
\autoref{lem:expectation_V} shows that $\sup_{i \geq 0}
\E\bV_{\beta_0}(\bX_i) < \infty$, and the proof of
\eqref{eq:martingale-noise} is concluded under the step size condition
\A{a:step-size-v2}.

Now consider \eqref{eq:reminder}. Decompose $\sum_{i=n}^k \gamma_{i,1}
\kappa_i^{(\ell)}=R^{(\ell,1)}_{n,k}+R^{(\ell,2)}_{n,k}+R^{(\ell,3)}_{n,k}$
with
\begin{align*}
&R^{(\ell,1)}_{n,k}\eqdef\sum_{i=n-1}^{k-1}\gamma_{i+1,1}\left[\bP_{(\bcovmat_{i},\brho_{i})} \hat{H}^{(\ell)}_{\bcovmat_{i},\brho_{i}}(\bX_{i}) -
\bP_{(\bcovmat_{i-1},\brho_{i-1})} \hat{H}^{(\ell)}_{\bcovmat_{i-1},\brho_{i-1}}(\bX_{i})\right]\\
&R^{(\ell,2)}_{n,k}
  \eqdef\gamma_{n-1,1}\bP_{(\bcovmat_{n-2},\brho_{n-2})} 
     \hat{H}^{(\ell)}_{\bcovmat_{n-2},\brho_{n-2}}(\bX_{n-1})
  -\gamma_{k,1}\bP_{(\bcovmat_{k-1},\brho_{k-1})} \hat{H}^{(\ell)}_{\bcovmat_{k-1},\brho_{k-1}}(\bX_{k})\\
&R^{(\ell,3)}_{n,k}\eqdef\sum_{i=n-1}^{k-1}(\gamma_{i+1,1}-\gamma_{i,1})\bP_{(\bcovmat_{i-1},\brho_{i-1})}
\hat{H}^{(\ell)}_{\bcovmat_{i-1},\brho_{i-1}}(\bX_{i})
\eqsp .
\end{align*}

By \autoref{lem:continuity-Poisson}  we get $|R^{(\ell,1)}_{n,k}|
\leq R^{(\ell,1,1)}_{n} + R^{(\ell,1,2)}_{n}$ where
\begin{align*}
R^{(\ell,1,1)}_{n} 
  &\eqdef K \sum_{i=n-1}^\infty \gamma_{i+1,1}\bV_{\beta_0}^{1/2}(\bX_i) 
  \big\|H^{(\ell)}_{\chunk{\rho_i}{1}{\ell}}
       -H^{(\ell)}_{\chunk{\rho_{i-1}}{1}{\ell}}\big\|_{\bV_{\beta_0}^{1/2}}\\
R^{(\ell,1,2)}_{n} &\eqdef \sum_{i=n-1}^\infty \gamma_{i+1,1}\bV_{\beta_0}^{1/2}(\bX_i) D_{\bV_{\beta_0}^{1/2}}\left[(\bcovmat_{i},\brho_{i}),(\bcovmat_{i-1},\brho_{i-1})\right]\eqsp.
\end{align*}
Since $H^{(\ell)}$ is locally Lipschitz with respect to $\brho$, by \autoref{lem:proj-distance} we obtain
\[
    \big\|H^{(\ell)}_{\chunk{\rho_i}{1}{\ell}}
    -H^{(\ell)}_{\chunk{\rho_{i-1}}{1}{\ell}}\big\|_{\bV_{\beta_0}^{1/2}}\leq K\gamma_{i,1}
    \eqsp.
\]
\autoref{lem:expectation_V} and \A{a:step-size-v2} imply
$\lim_{n\to \infty} R^{(\ell,1,1)}_{n} = 0$ \as\ 
and \autoref{lem:series-DV} shows that 
$\lim_{n\to \infty} R^{(\ell,1,2)}_{n} = 0$ \as. 
By \eqref{eq:inequality-poisson2}, \autoref{lem:expectation_V}
and \A{a:step-size-v2}, the sum 
$\sum_k \big(\gamma_{k,1}\bP_{(\bcovmat_{k-1},\brho_{k-1})} 
\hat{H}^{(\ell)}_{\bcovmat_{k-1},\brho_{k-1}}(\bX_{k})\big)^2<\infty$
\as, implying that 
$\sup_{k\geq n}|R^{(\ell,2)}_{n,k}|\stackrel{n\to\infty}{\rightarrow}0$
\as.
Finally, for the term $R^3_{n,k}$, \eqref{eq:inequality-poisson}, 
\autoref{lem:expectation_V} and \A{a:step-size-v2} conclude the proof.
\end{proof} 

\begin{proof}[Proof of \autoref{theo:convergence-temp-adapt}] 
 For any $\ell=1,\dots,L-1$, define  the local Lyapunov function
 $w^{(\ell)}:\R\to\R$ through
\begin{equation}\label{definition-lyapunov-function}
    w^{(\ell)}(\rho)\defeq (\rho-\hat{\rho}^{(\ell)})^2\eqsp,
\end{equation}
where $\hat{\rho}^{(\ell)}$ is the unique solution of the equation
$g^{(\ell)}(\rho)= 0$ where $g^{(\ell)}$ is defined in
\eqref{eq:definition-g}. Observe that, by \autoref{lem:monotonic:h}
and \autoref{prop:uniqueness-temperature}, for any $\ell=1,\dots,L-1$
the functions $w^{(\ell)}$ satisfy, for all $\rho \in \rset$
\[
\frac{\partial w^{(\ell)}}{\partial \rho}(\rho) \eqsp \cdot \eqsp g^{(\ell)}(\rho) \leq 0\eqsp,
\]
with equality if and only if $\rho=\hat{\rho}^{(\ell)}$. The projection
set $\ccint{\underline{\rho},\overline{\rho}}$ obviously satisfies
\cite[A4.3.1]{Kushner2003} and since for all $\ell=1,\dots,L-1,$ 
$\underline{\rho}<\hat{\rho}^{(\ell)}<\overline{\rho}$, the
convergence cannot occur on the boundary of constraint set. Hence we
only need to check the assumptions of \cite[Theorem~6.1.1]{Kushner2003}.
Since the mean field $g^{(\ell)}$ is bounded and continuous, and depends only
on $\rho$ the conditions \cite[A6.1.2, A6.1.6, A6.1.7]{Kushner2003}
are satisfied. Condition \cite[A6.1.1]{Kushner2003} is implied by the
boundedness of $H^{(\ell)}$. Condition \cite[A6.1.3]{Kushner2003} is
implied by \autoref{prop:noise}. 

Once we show that the remainder
term $r_n^{(\ell)}$, defined in \eqref{eq:definition-r} converges \as\
to zero as $n\to\infty$, \autoref{prop:noise} implies the condition
\cite[A6.1.4]{Kushner2003}. We will prove inductively 
that 
\begin{equation}\label{eq:conv-r}
\lim_{n\to\infty}|r^{(\ell)}_n|=0\quad\as\eqsp.
\end{equation}
Consider first the case $\ell=1$. By \eqref{def:temp_adapt_second} we
get that $r^{(1)}_n=0$. Assume that for $\ell<k$ \eqref{eq:conv-r}
holds, then by \cite[Theorem~6.1.1]{Kushner2003} for any $\ell<k$ we
know that $\rho_n^{(\ell)}\to\hat{\rho}^{(\ell)}$. To simplify
notation we denote $u_n=u(\rho_n^{(k)})$,
$v_n=v^{(k)}(\chunk{\rho_n}{1}{k-1})$ and
$\hat{v}=v^{(k)}(\chunk{\hat{\rho}}{1}{k-1})$.  By
\eqref{eq:definition-fv}, we get that $f_u(z)\leq 2$ for all $z$.
Hence, using \eqref{eq:definition-tildeh}, we get
\begin{align*}
|r_n^{(k)}|&\leq |\tilde{h}(u_nv_n,v_n)-\tilde{h}(u_n\hat{v},v_n)|+|\tilde{h}(u_n\hat{v},v_n)-\tilde{h}(u_n\hat{v},\hat{v})|\\
&\leq\left|\int f_{v_n}(\pi(y))\left[\frac{\pi^{u_nv_n}(y)}{Z(u_nv_n)}
-\frac{\pi^{u_n\hat{v}}(y)}{Z(u_n\hat{v})}\right]\ud y\right| \\&\quad+\left|\int f_{u_n\hat{v}}(\pi(y))\left[\frac{\pi^{v_n}(y)}{Z(v_n)}
-\frac{\pi^{\hat{v}}(y)}{Z(\hat{v})}\right]\ud y\right|\\
&\leq 2\int \left|\frac{\pi^{u_nv_n}(y)}{Z(u_nv_n)}
   -\frac{\pi^{u_n\hat{v}}(y)}{Z(u_n\hat{v})}\right| \ud y +2\int \left|\frac{\pi^{v_n}(y)}{Z(v_n)}
   -\frac{\pi^{\hat{v}}(y)}{Z(\hat{v})}\right| \ud y \\
&= 2\left\{\tvnorm{\frac{\pi^{u_nv_n}}{Z(u_nv_n)}-\frac{\pi^{\hat{u_nv}}}{Z(\hat{u_nv})}}+\tvnorm{\frac{\pi^{v_n}}{Z(v_n)}-\frac{\pi^{\hat{v}}}{Z(\hat{v})}}\right\}
\eqsp .
\end{align*}
By \cite[Lemma A.1]{Douc2002} we get that 
\[
|r_n^{(k)}|\leq 4\left\{|\log(Z(u_nv_n))-\log(Z(u_n\hat v))|+|\log(Z(v_n))-\log(Z(\hat v))|\right\}\eqsp,
\] where the normalising constants $Z$
are defined in \eqref{eq:definition-normalizing-constant}.  Since
$\alpha\mapsto \log(Z(\alpha))$ is locally Lipschitz, $u_n$ and $v_n$ are in compact set we get that exists $K<\infty$ such that
\[|r_n^{(k)}|\leq K(1+u_n)|v_n-\hat v |
\eqsp.\]
By \eqref{def:temp_adapt} and \eqref{eq:definition-u} we obtain that $\sup_n |u_n|<\infty$. Hence, since
$\chunk{\rho}{1}{\ell-1}\mapsto v^{(k)}(\chunk{\rho}{1}{\ell-1})$ is continuous, by induction assumption,
we conclude that $|r_n^{(k)}|\to 0\ \as$.
\end{proof} 


\section*{Acknowledgements}

The first two authors were supported by the French National Research Agency under the contract ANR-08-BLAN-0218 
``BigMC''.
The third author was supported by the Academy of Finland
(project 250575) and by the Finnish Academy of Science and Letters, Vilho,
Yrjö and Kalle Väisälä Foundation.


\end{document}